\documentclass{amsart}

\usepackage{amsmath}
\usepackage{amssymb}
\usepackage{amscd}
\usepackage{mathtools}
\usepackage[utf8]{inputenc}
\usepackage[T1]{fontenc}

\def\dj{d\kern-.30em\raise1.25ex\vbox{\hrule width .3em height .03em}}
\def\Dj{D\kern-.70em\raise0.75ex\vbox{\hrule width .3em height .03em}
\kern.03em}

\makeatletter
\renewcommand{\subsection}{\@startsection{subsection}{2}{\z@}%
{\baselineskip}{0.5\baselineskip}{\bfseries}}
\def\l@section{\def\@tocpagenum##1{\hss{\bfseries ##1}}
\@tocline{1}{8pt}{0pc}{}{\bfseries}}
\def\l@subsection{\def\@tocpagenum##1{}
\@tocline{2}{2pt}{2pc}{2pc}{}}
\makeatother

\def\qqqq#1#2#3#4{\prescript{#1}{#3}{\lozenge}^{#2}_{#4}}
\def\braket#1#2{\langle #1\vert #2\rangle}
\def\bra#1{\langle #1\vert}
\def\ket#1{\vert #1\rangle}
\def\braketop#1#2#3{\langle #1 \vert #2\vert #3\rangle}
\def\diada#1#2{\vert #1\rangle\langle #2\vert}
\def\lqbra{\raisebox{1.2pt}{$\scriptscriptstyle\blacktriangleleft$} \kern-2.8pt\langle }
\def\rqbra{\rangle\kern-3pt\raisebox{1.2pt}{$\scriptscriptstyle\blacktriangleright$}}
\def\qform#1#2{\lqbra #1\vert #2\rqbra}
\def\polq{\Bbb{C}[z,\bar{z},Q]}
\def\azz{\cal{A}}
\def\pp{\phi}
\def\zz{z^*}

\def\Cz{\Bbb{C}[z]}
\def\Czf#1{\Bbb{C}_{#1}[z]}
\def\Opiz{\mathrm{O}(\Bbb{C})}
\def\Opz{\mathrm{o}(\Bbb{C})}
\def\HolC{\mathrm{H}(\Bbb{C})}
\def\W{\cal{W}}
\def\zzinfty{\reflectbox{$z$}}
\def\pointf#1#2{[#1{	\wr}#2]}
\def\Tz{\mathbf{T}}
\def\de{\mathrm{d}}

\def\pinfty{\cal{Q}}
\def\sinfty{S_\infty}
\def\ninfty{\cal{N}}
\def\qinfty{\cal{R}}

\def\cal{\mathcal}
\def\Bbb{\mathbb}

\newcommand{\id}{\mathrm{id}}

\newcommand{\im}{\mathrm{im}}

\newtheorem{lemma}{Lemma}[section]

\newtheorem{prop}{Proposition}[section]
\theoremstyle{definition}
\newtheorem{definition}{Definition}[section]
\newtheorem{remark}[definition]{Remark}

\numberwithin{equation}{section}

\begin{document}

\renewcommand{\thepage}{\ifnum\value{page}=1\else\arabic{page}\fi}

\title[Coherent States and Toeplitz Quantization]
{Hilbert Spaces of Entire Functions and \\ 	
Toeplitz Quantization of Euclidean Planes}

\author{Micho \Dj ur\dj evich}

\address{Universidad Nacional Aut\'onoma de M\'exico, 
Instituto de Matem\'aticas, Area de la Investigacion Cient\'{\i}fica, 
Circuito Exterior, Ciudad Universitaria, CP 04510,
Mexico City, Mexico}
\email{micho@matem.unam.mx}

\author{Stephen Bruce Sontz}

\address{Centro de Investigaci\'on en Matem\'aticas, A.C., 
(CIMAT), Jalisco S/N, Mineral de Valenciana, CP 36023, 
Guanajuato, Mexico}
\email{sontz@cimat.mx}

\vspace{-5pt}
\begin{abstract} 
The theory of Toeplitz quantization presented in our previous paper is extended and further developed 
to include diverse and interesting non-commutative 
realizations of the classical Euclidean plane. 
This is done using Hilbert spaces of entire 
functions, where polynomials in one complex variable form a dense subspace. 
The complex coordinate naturally acts as an unbounded multiplication operator generating, together with its adjoint, 
a highly non-commutative *-algebra of operators. 
The Toeplitz operators are then geometrically constructed as special elements from this algebra;   they are associated to the symbols from another quadratic non-commutative algebra, which is interpretable as polynomials over a plane to be quantized. 
Such a conceptual framework promotes interesting non-trivial conditions on the initial scalar product. 
These are analyzed 
in detail. 
Various illustrative examples are computed.     
\end{abstract}
\maketitle

\vspace{-10pt}
\tableofcontents

\section{Introduction}

\parskip=2pt plus 2pt minus 2pt

This paper is a continuation of our previous work \cite{coherent}, where we have provided a detailed analysis of geometric, algebraic and analytic aspects for a coherent states quantization of the Manin $q$-plane via Toeplitz operators.  
A general method for such Toeplitz quantizations has been introduced and extensively studied 
by the second author in \cite{sbs1,sbs2}. 

Here we would like to further expand on these principal dimensions by exploring their deeper interrelations and   
establishing a link with the Hilbert spaces of entire functions \cite{barg, debranges}. 
On one hand this allows the use of a rich array of techniques from complex analysis, and on the other  it provides an elegant framework for defining quantum versions of classical fundamental
geometrical objects---such as points, coordinates, observables and transformations. 

Our context fits well into the principal quantum  non-commutative geometry foundations as developed by Connes   
\cite{connes} and Woronowicz \cite{slw-cat} with a variety of interesting non-commutative algebras, representable 
by operators that appear naturally in a Hilbert space. 
And moreover, 
our approach is particularly in resonance with the pioneering formulation of quantum geometry by Prugove\v{c}ki \cite{pru}, where the concept of point is quantized by morphing it into a special wave function. 
These `wave functions of points' can then be interpreted as coherent states, and in our framework they 
are an integral part of the Hilbert space geometry, which directly emerges from the reproducing kernel. 

The reproducing kernel is a primary and defining structural part of Hilbert spaces of holomorphic functions 
defined over an open domain $\Omega \subset \Bbb{C} $. 
This is a map $K\colon \Omega\times\Omega\rightarrow \Bbb{C}$, antiholomorphic in the first coordinate and holomorphic in the second and exhibiting an appropriate matrix complete positivity condition. Such a structure gives us the possibility 
of speaking of  `quantized' points, these being 
the vectors of the Hilbert space representing the points of $\Omega$. 
This association is obtained from the kernel map by $\Omega\ni w\longmapsto K(\bar{w},*)\in\cal{H}$. 
By taking scalar products
these {\em point wave functions}  evaluate in the points $w$ the functions of $\cal{H}$. 
And by 
normalizing them we obtain the corresponding coherent states. We shall here almost exclusively deal with 
entire functions, which 
correspond to taking the domain $\Omega=\Bbb{C}$. 

Let us outline the contents of the paper. 
In the next section in a series of steps
we establish  
our principal geometric and algebraic setup. 
We begin with a detailed analysis of scalar products in the space $\Cz$ of complex polynomials in 
a complex variable $z$. 
Every such 
scalar product $\braket{}{}$ leads to a Hilbert space $\cal{H}$, namely the completion of $\Cz$ with respect to 
$\braket{}{}$.  
So $\Cz$ is dense in $\cal{H}$. 
We would like to capture in terms of the operators on $\cal{H}$ some basic idea of a plane, 
be it classical or quantum. This naturally leads to interpreting the complex coordinate $z$ as an appropriate 
multiplication operator $Z$ in $\cal{H}$, so that on polynomials $ \pp $ it acts as $Z\colon \pp(z)\mapsto z\pp(z)$. 
We would also like to be able to consider the operator representation of the conjugate coordinate $\bar{z}$ 
as an adjoint operator $Z^*$ so that an interesting operator calculus involving $Z$ and $Z^*$ 
can be defined, 
reflecting the geometrical idea of an underlying plane-like space. 
An important technical problem here is to find the most effective context for defining $Z$ and $Z^*$, 
including their domains. 
From consideration of this problem  
a number of non-trivial conditions for the initial scalar product on $\Cz$ 
naturally emerge. We shall call them {\em Harmony Properties}, and we shall number 
them from $0$ to $3$. 

Harmony Zero treats the very basic structure at the level of polynomials, while Harmony One provides a particularly elegant and geometrically natural answer by postulating  
that all the elements of the completion $\cal{H}$ of $\Cz$ are interpretable as entire functions. The definition of the operator $Z$ then can be extended from the polynomials to its natural domain, in the same way as the complex 
coordinate multiplication operator $Z\colon\psi(z)\mapsto z\psi(z)$ with dense domain  $\mathrm{D}(Z)$ consisting of all $\psi(z)\in\cal{H}$ for which $z\psi(z)\in\cal{H}$. We shall establish a variety of important 
properties for the operator $Z$ and its adjoint $Z^*$. 
In particular $Z$ is closed and the point wave functions associated to classical points $w\in\Bbb{C}$ are the unique (modulo non-zero scalar multiples) 
eigenvectors for $Z^*$ with the eigenvalues $\bar{w}$. 
This implies that the spectra of $Z$ and $Z^*$ coincide with the whole $\Bbb{C}$; so they are unbounded operators. 

However, this harmony property still does not ensure the  
existence of an algebraically effective operational setting involving the operators $Z$ and $Z^*$ so that we can 
meaningfully compose them in arbitrary combinations.

Our next harmony property addresses this, by postulating the continuity of $Z^*$ in the normal topology 
of $\cal{H}$ 
(namely, the topology of uniform convergence on
compact sets)
and thereby putting $Z$ and $Z^*$ in a kind of balanced relationship. As we shall explain, this
provides a common acting space $\W$ for them, consisting of special 
vectors emerging from the very geometry of $\cal{H}$. 
The point wave functions are always in $\W$, and for polynomials to be included, we shall need the topological 
half of Harmony Zero. 

In such a framework the operators $Z$ and $Z^*$ generate a fundamental *-algebra $\azz$, in terms of which the underlying quantum 
space is brought to life. We shall prove that $\azz$ is quite non-commutative, as it 
is always with a trivial centralizer in a pretty large algebra $\Opiz$ of normally continuous linear transformations 
of entire functions. We shall also prove that, under certain 
general additional assumptions on interchangeability of order of the product of $Z^*$ and $Z$, the underlying quantum space is always `pointless'. 
In other words there will be no characters on $\azz$. 

We shall then proceed with interpreting the constructed framework as a 
natural destination for a Toeplitz quantization of appropriate quadratic algebras $\polq$. These algebras are generated by coordinates $z$ and $\bar{z}$, and equipped with a quadratic flip-over type relation $Q$ between these coordinates. This relation can be trivial, in which 
case we are dealing with the classical complex plane $\Bbb{C}$, twisted commutative as in the case of the  Heisenberg-Weyl algebra or the Manin $q$-plane, 
or a more elaborate quadratic relation telling us how to exchange 
$z$ and $\bar{z}$. 

The main construction here is that of an extended space $\qinfty$ equipped with a non-degenerate, 
but not necessarily positive definite form $\qform{}{}$, such that the common action space $\cal{W}$ for the operators of $\azz$ together 
with its scalar product $\braket{}{}$ can be singled out via a symmetric idempotent $\Pi$ acting on $\qinfty$, and such that the algebra $\polq$ acts symmetrically on $\qinfty$ with $\cal{W}$ being the cyclic $Z$-invariant subspace. 
The Toeplitz operator with the symbol $f\in\polq$ is then defined as $\Pi f\Pi$ interpreted as an element of $\azz$. 
 
It is worth mentioning that the principal algebraic steps in a sense mirror those of the Stinespring
construction \cite{Sti} for completely positive maps between C*-algebras. As we shall see in detail, however, we are in the context of unbounded operators and a key positivity condition is not always fulfilled. On the other hand 
when $\qform{}{}$ 
is positive -- and this will be our final harmony condition -- we can close $\qinfty$ in a Hilbert space $\cal{J}$ and the relationship between $\polq$ and $\azz$ established by the Toeplitz quantization turns out to have 
a particular richness. 

We then proceed with some more detailed calculations of a number of concrete examples by analyzing the coherent states, their resolution of the identity and the Toeplitz quantization interpretation. 
We next analyze four  principal 
types for the algebras $\polq$, and we also  discuss their possible realizations in  Hilbert spaces of analytic 
(but not necessarily entire) functions. 
As we shall see, the Heisenberg-Weyl algebra and its $q$-variations are 
the only types allowing a realization with entire functions. 
On the other hand the Manin $q$-plane
is naturally realizable in a space of Laurent series, 
whose common domain 
of definition 
is $\Bbb{C}-\{0\}$. Some of these 
spaces have the unit disk or the half plane as their common domain of definition for the holomorphic functions, and 
these provide a natural setting for quantizing the classical 
hyperbolic plane \cite{MP}.  
In all these calculations  basic identities involving hypergeometric series 
and $q$-special functions (such as the triple product identity, the $q$-binomial theorem, the second Euler identity, $q$-versions of exponential, logarithm and Gamma functions) naturally appear. The classical treaties 
\cite{SF} and \cite{BHS} provide an excellent in-depth exposition of these identities. 

In Section~3 we focus on some general algebraic aspects of Toeplitz quantization. We explain how the construction of the extended algebra and module with the projector and the quadratic form 
can be performed over 
a large class of algebras $\cal{C}$ equipped with a flip-over operator $\sigma\colon\cal{C}\otimes\bar{\cal{C}}\rightarrow\bar{\cal{C}}\otimes\cal{C}$, where $\bar{\cal{C}}$ is the conjugate algebra of $\cal{C}$. 
This includes as a special case our main protagonist---the polynomials $\Cz$---and also multi-dimensional versions with several complex variables. 

Finally, in Section~\ref{conclusions} some concluding remarks are made. The paper ends with two extensions exhibiting perhaps some interest on their own. 

In the first appendix, we 
provide a simple geometrical interpretation for the positivity of the canonical quadratic form 
and of the Stieltjes moment condition. 
And in the second appendix 
some formulas and illustrative examples of Hilbert spaces of holomorphic functions are collected, including the 
classical structures such as Segal-Bargmann  spaces and  their $q$-variations, de Branges spaces including the Paley-Wiener space, and the Bergman spaces of square integrable holomorphic functions over bounded domains. For all these spaces a similar quantization scheme for their underlying classical domain can be established.  

We use standard notation such as $ \Bbb{N} $ 
for the set of non-negative integers and 
$ \Bbb{C} $ for the complex plane. 

\section{Quantization Via Hilbert Spaces of Entire Functions}\label{setting-section}

\subsection*{Polynomials Generating a Hilbert Space}

The central object for our considerations will be the polynomial algebra $\Cz $ in one complex variable $z$. It is a complex infinite dimensional vector space, in which the monomials $\{z^n \,|\,n\in\Bbb{N}\}$ form a natural linear basis. 
For every $n\in\Bbb{N}$ the polynomials 
of degree $\leq n$ form an $(n+1)$-dimensional 
subspace, denoted by $\Czf{n}$.  

The polynomials $\Cz$ are included as a subalgebra in  
the larger algebra of all entire functions,
denoted by $\HolC$. 
In the normal topology we have 
\begin{equation}
\overline{\Cz}=\HolC, 
\end{equation} 
where $ \overline{A} $ denotes the 
topological closure of $ A $. 

Non-constant polynomials can also be understood as entire functions $\pp\colon
\Bbb{C}\rightarrow\Bbb{C}$ preserving the complex infinity point $\infty$. 
Under composition polynomials form a non-commutative unital semi-group. 
An algebraic counterpart to this geometrical view is the interpretation of polynomials as 
unital algebra 
homomorphisms $\pp\colon\Cz\rightarrow\Cz$. 
This linear action of $\pp$ on $\Cz$,
which is called {\em substitution}, is 
specified on basis elements 
by $\Cz\ni z^n\longmapsto \pp(z)^n\in\Cz$. 

Of particular interest are 
those polynomials whose corresponding homomorphism 
is invertible, which we call the 
{\it invertible polynomials}. 
They are precisely the linear transformations
$z\longmapsto az+b$
where $a,b\in\Bbb{C}$ with $a\neq 0$. They provide all the orientation preserving symmetries of the Euclidean plane $\Bbb{C}$. If $a=1$ we have translations given by $b$ 
(with no fixed point if $ b \ne 0 $)
and {\em amplitwists} around the unique fixed point 
$b/(1-a)$ if $a\neq 1$ with rotation factor given by the phase of $a$ and similarity factor given by $\vert a\vert$. 
(We continue assuming $ a\neq 0 $.)
It is well known that 
all the automorphisms of the algebra $\Cz $ are of this form, and this is the same as all the 
holomorphic automorphisms of the complex plane $\Bbb{C}$. 

There is a canonical anti-linear involutive algebra
automorphism $J$ of $\Cz $ defined by $J(z)=z$ 
and which just conjugates the 
coefficients of the polynomials. 

The multiplication in $\Cz$ can be interpreted as the left regular representation, where $\Cz$ acts by linear 
operators in $\Cz$. The 
constant polynomial $1\in\Cz$ is both cyclic and separating for this representation. Viewed in terms of the 
left regular representation, $\Cz$ is a {\it maximal commutative subalgebra} of linear operators in $\Cz$: any linear operator 
$l:\Cz \to \Cz$ commuting with 
left multiplication by $z$ (and hence by induction 
and linearity with left 
multiplication by any polynomial)  
is left multiplication by $l(1)\in\Cz$.  

Among other important operators acting in $\Cz$ is the complex differentiation $\partial/\partial z$, and hence all 
differential operators whose  coefficients
are polynomial in $ z $. 
All these operators---multiplication, substitution, 
differentiation---are continuous 
(which will always mean in this context 
with respect to the normal topology
on $\Cz$ induced by its inclusion 
in $\HolC$),   
and they naturally extend to the whole 
$\HolC$. 
We shall denote by $\Opz$ the algebra of all 
normal 
continuous linear transformations of $\Cz$ to itself: 
$$
\Opz := \{ \tau : \Cz \to \Cz \,\big|\, \tau 
\mathrm{~is~linear~and~normal~continuous} \}. 
$$
Every $\tau
\in\Opz$ is uniquely extendible to a normal 
 continuous linear transformation $\tau\colon\HolC\rightarrow\HolC$. We can say 
equivalently that $\Opz$ consists of continuous linear transformations of $\HolC$, which preserve $\Cz$.  
So $\Opz$ is a subalgebra of the algebra of all continuous linear transformations of $\HolC$, which will be denoted 
by $\Opiz$: 
$$
\Opiz := \{  \sigma : \HolC \to \HolC \,\big|\, \sigma 
\mathrm{~is~linear~and~normal~continuous}  \}. 
$$
The algebra $\HolC$ acts on itself 
via its left regular 
representation in terms of which $\HolC$ is a maximal commutative subalgebra of $\Opiz$.  

The dual space $\HolC^*$ consists by definition of all normally continuous complex linear functionals  
$ f : \HolC \rightarrow \Bbb{C}$. 
This is the same as saying normally continuous functionals on $\Cz$, 
since each of the latter uniquely 
extends by normal continuity to $\HolC$. 
The elements $f\in\HolC^*$ are in one-to-one correspondence with the complex sequences 
$f_n=f(z^n)$ satisfying the following {\it geometric boundedness property}: there exists a constant $\Lambda>1$ such that for all $n\in\Bbb{N}$ we have 
\begin{equation}\label{geometric-boundedness}
|f_n|\leq\Lambda^n. 
\end{equation}

In addition to this very basic structure, 
we shall assume that $\Cz$ is equipped 
with a positive definite scalar product $\braket{}{}$, 
which is anti-linear in its first entry 
and linear in the second. 
It is completely specified by the infinite matrix 
\begin{equation}
\sinfty=(s_{nm})\quad \mathrm{where} \quad s_{nm}=\braket{z^n}{z^m}
\quad \mathrm{for~} n,m \in \mathbb{N}. 
\end{equation}
This matrix is required to be
be hermitian and moreover to be strictly positive 
in the sense that the partial matrices
\begin{equation}
S_n=\begin{pmatrix}s_{00} & \cdots & s_{0n}\\
\vdots & \ddots & \vdots\\
s_{n0} & \cdots & s_{nn} 
\end{pmatrix}
\end{equation}
determining the scalar product in 
the subspaces $\Czf{n}$ are strictly positive 
matrices for every $n\in\Bbb{N}$. 
For a hermitian matrix 
this will be the case if and only if all the numbers 
\begin{equation}
d_n=\det(S_n)
\end{equation}
are strictly positive real numbers. 
In particular, $S_n$ is invertible and $S_n^{-1}$ is  strictly positive, too.  

\begin{remark}
All our constructions will be independent of rescalings of the scalar product. So we can always 
normalize $\braket{}{}$ by fixing
a value for $s_{00}>0$, for example by postulating $s_{00}=1$.  
\end{remark}

The following is a straightforward result. 

\begin{prop} The involution $J$ is an isometry if 
and only if all the coefficients $s_{nm}$ are real. \qed
\end{prop}

By using the Gram-Schmidt procedure
to orthonormalize the naturally ordered sequence of 
monomials $z^n$,    
we obtain an orthonormal sequence of polynomials $\pp_n(z)$ such that 
the degree of $\pp_n(z)$ is $n$ with 
its highest order coefficient being positive. 
Also $\pp_{n+1}(z)\bot \Czf{n}$. 
The space $\Czf{n}$ is spanned by $\pp_0(z), \dots, \pp_n(z)$, and 
\begin{equation}
\label{B-basis-defined}
B:= \Big\{\pp_n(z)  ~\Big|~ n \ge 0  \Big\} 
\end{equation}
is an orthonormal vector space basis 
in the whole space $\Cz$. 

Explicitly, these canonical orthonormal polynomials are given by 
\begin{equation}
\pp_n(z)=\frac{1}{\sqrt{d_nd_{n-1}}}\left | \begin{matrix} s_{00}& s_{01} & \cdots & s_{0n} \\
\vdots & \vdots & \ddots & \vdots \\
\,s_{n-1 0} & s_{n-1 1} &\cdots & s_{n-1 n}\,\\
1 & z& \cdots & z^n
\end{matrix}\right |
\end{equation}
with the extra definition $d_{-1}:=1$. 

\begin{remark}
When appropriate we shall also write simply 
$\ket{n}=\pp_n(z)$. 
\end{remark}

The following polynomial identity holds: 
\begin{equation}\label{n-kernel}
\sum_{k=0}^n\overline{\pp_n(w)}\pp_n(z)=\sum_{i,j=0}^n z^i[S_n^{-1}]_{ij}\bar{w}^j.
\end{equation}

The correspondence between the scalar products $\braket{}{}$ on $\Cz$ 
and systems $\pp_n(z)$ is in fact 
one-to-one. If a system $\pp_n(z)$ of polynomials satisfying $\deg \pp_n(z)=n$ and with positive highest-order coefficients is given, then there exists a unique scalar product $\braket{}{}$ on $\Cz$ orthonormalizing the monomials $z^n$ into $\pp_n(z)$. 

Let us check how the scalar product matrix $\sinfty$
transforms under holomorphic symmetries of $\Bbb{C}$.
 Again these are precisely the linear maps 
$z\longmapsto az+b$ with $a,b\in\Bbb{C}$ and
$a\neq 0$. 
We shall 
consider separately rotations, scalings and translations. 

First, under the rotations $z\longmapsto uz$ around $0$ by a unitary complex number 
$u$ the coefficients $s_{nm}$ transform as 
\begin{equation}\label{u-transform}
s_{nm}\rightsquigarrow u^{m-n}s_{nm}. 
\end{equation}
In particular the scalar product will be invariant under all of these transformations 
if and only if $s_{nm}=0$ for all $n\neq m$. 
Second, 
the scaling $z\longmapsto r z$ with $r>0$ gives 
\begin{equation}\label{r-transform}
s_{nm}\rightsquigarrow s_{nm}r^{n+m}. 
\end{equation}
Finally, the transformation rule for translations $z\longmapsto z+b$ is 
\begin{equation}\label{b-transform}
s_{nm}\rightsquigarrow \sum_{k=0}^n\sum_{l=0}^m
\binom{n}{k}\binom{m}{l}\bar{b}^{n-k}b^{m-l}s_{kl}.
\end{equation}

The transformations \eqref{u-transform}, 
\eqref{r-transform} and \eqref{b-transform} are 
easy to verify. 
From this it is also easy to see that every transformation which is not a rotation around some point of $\Bbb{C}$ 
affects non-trivially the scalar product. 

\begin{remark}
If the scalar product is 
invariant under the rotations \eqref{u-transform} 
($\Leftrightarrow$ 
the matrix $\sinfty$ is 
diagonal), we simplify the notation by 
defining a sequence 
\begin{equation}
\label{simple-sequence}
s_n:=s_{nn} \quad \mathrm{for~} n \in \mathbb{N}.
\end{equation}
The canonical 
orthonormal polynomials are then simply 
\begin{equation}\label{pp-n-s}
\pp_n(z)=\frac{z^n}{\sqrt{\smash[b]{s_n}}} \quad \mathrm{for~} n \in \mathbb{N}. 
\end{equation}
In our previous paper \cite{coherent}, 
we only dealt with such scalar products. 
\end{remark}

Let $\cal{H}$ be the Hilbert space obtained 
by completing the incomplete 
pre-Hilbert space 
$\Cz$ relative to $\braket{}{}$. 
Then the subspace $\Cz$
is dense in $\cal{H}$, and 
the set $B$ in \eqref{B-basis-defined} 
is an orthonormal basis for $\cal{H}$. 
Of course,  $\cal{H}$ is not unique, but 
it is unique up to a unique isometric 
isomorphism which is the identity on $\Cz$. 
One possible explicit construction of $\cal{H}$
is as the set of all formal infinite series 
$ \sum_{n \in \Bbb{N} } a_{n} \pp_n(z) $ with 
coefficients $ a_{n} \in \Bbb{C} $ satisfying
$ \sum_{n \in \Bbb{N} } |a_{n}|^{2} < \infty $. 
In this setting the dense subspace $ \Cz $ 
is identified as the set of all formal sums for which 
only finitely many coefficients are non-zero. 
This clearly shows that $ \Cz $ is not equal 
to $\cal{H}$. 
Other realizations of $\cal{H}$ will be considered 
later. 

There exist four important, mutually subtly related, conditions for the scalar product $\braket{}{}$ and 
consequently for 
the resulting Hilbert space $\cal{H}$. 
These conditions  
put this simple framework into special harmony with operator theory and complex analysis, thereby 
establishing an elegant context 
for constructing quantum models of the complex 
Euclidean plane $ \Bbb{C} $.  
We shall call them 
Harmony Properties and assign them numbers $0$, $1$, $2$ and $3$. 

Harmony Zero is a dual thing, namely 
a symbiosis of both an 
algebraic and a topological condition. 
It provides the simplest, most 
straightforward *-algebraic structure emerging 
from the space of polynomials and their scalar product 
$\braket{}{}$, and at the same 
time it maintains an elementary topological resonance with complex analysis. 

But first we shall say that an operator 
$ S : \Cz \to \Cz $ is {\it formally adjointable} 
with respect to $ \braket{}{} $ if there 
exists (a necessarily unique) operator 
$ T : \Cz \to \Cz $ satisfying 
$ \braket{\phi}{S \psi} = \braket{T\phi}{\psi} $
for all $ \phi, \psi \in \Cz $, in which 
case we say that $ T $ is the {\it formal adjoint} of $ S $. 
Also, for $ \pp, \psi \in \Cz $
we recall the Dirac notation $\bra{\pp}$ for 
$ \psi\rightarrow \braket{\pp}{\psi}$, a linear 
functional on $ \Cz $, which is continuous in the 
$ \braket{}{} $ (or norm) topology on $ \Cz $. 

\begin{definition} We shall say that 
{\em Harmony Zero} (or simply H0) holds if: 
\begin{itemize}
\item The multiplication operator by $ z $ is formally adjointable in $\Cz$ with respect to the 
inner product $\braket{}{}$. 

\item For every $ \pp \in \Cz $
the linear functional $\bra{\pp}$ is continuous in the normal topology of $\Cz$. 
\end{itemize}
We shall refer to the first condition as Algebraic H0 and to the second condition as Topological H0. 
The reader is advised to carefully 
note the distinction between the norm topology 
and the normal topology. 
In the context of Algebraic H0 we let 
$ z $ denote the operator of 
multiplication 
$ z : \Cz \to \Cz $
and denote its formal 
adjoint by 
$ \zz : \Cz \to \Cz$.
The operator $\zz$ should not be confused 
with the conjugate complex variable $ \bar{z} $. 
\end{definition}

Here are some useful reformulations of these properties.
As for Analytical H0 
let us consider the infinite matrix $z_{nm}=\braketop{n}{z}{m}:= 
\braket{\phi_{n} (z)}{z \phi_{m} (z)}$.  
(More Dirac notation.)
This represents the multiplication operator 
by $z$ in the basis $B$. 
By the construction of $B$, we have that 
if $n-m\geq 2$, then $z_{nm}=0$. 
So each column of this matrix has only finitely many 
non-zero entries. 
Moreover, $z_{nm}> 0$ if $n-m=1$. 
The formal adjointability of $z$ is equivalent to the statement that every row of this matrix has only finitely many non-zero entries. 
If so, then 
the formal adjoint operator  
$\zz : \Cz \to \Cz$ 
to the multiplication operator 
by $z$ in $\Cz$ is represented in $B$ 
by the corresponding adjoint matrix. 

As for Topological H0 let us observe first that the normal 
continuity of all the dual vectors $\bra{\phi}$ 
is equivalent to the continuity of
only all the basis vectors, that is of 
$\bra{n} = \bra{\pp_n(z)}$ 
for all $ n \in \Bbb{N} $. 
In accordance with \eqref{geometric-boundedness} this  property can be expressed as 
saying that for every $ m \in \Bbb{N}  $ 
there exists a real number $ \Lambda_m $
such that we have 
the system of inequalities  
\begin{equation}\label{H0-topological}
|s_{mn}|\leq (\Lambda_m)^n 
\quad \mathrm{for~all~} n \in \Bbb{N} .    
\end{equation}
In particular this is the case 
if the rows ($\Leftrightarrow$ the columns) of $S_\infty$ possess only finitely many non-zero entries. 

\begin{remark}
The two components of H0, the algebraic and the 
topological, do not entangle strongly within this basic polynomial context. 
Their mutual correlations are manifested through the higher harmony levels.   
\end{remark}

\begin{remark} Harmony Zero always holds for diagonal scalar products. Indeed, in this case a quick 
	calculation shows that   
\begin{equation}\label{Z-up}
z \ket{n}=\left(\frac{s_{n+1}}{s_n}\right)^{1/2}\!\!\! \ket{n+1}\quad\quad \forall n\ge 0. 
\end{equation}
This is a 
kind of {\it creation} operator with the only non-zero matrix 
entries being those immediately below the main diagonal. 
It is formally adjointable in $\Cz$ 
and its formal adjoint, 
the {\it annihilation} operator, is given by 
\begin{equation}\label{Z-down}
   \zz \ket{n}= 
   \left( \dfrac{s_{n}}{s_{n-1}} \right)^{1/2}
   \!\!\! \ket{n-1}
   \quad\quad \forall n \ge 1,
\end{equation}
together with 
$ \zz \ket{0} =0 $. 
As for the topological part of H0 
that follows from the above 
observation on non-zero entries of $S_\infty$ 
and comparing with \eqref{H0-topological}. 
Or it follows explicitly from the formula
\begin{equation}
\braket{n}{\phi(z)}=\frac{s_n^{1/2}}{n!}\frac{\partial^n}{\partial z^n}\phi(z)\Bigm\vert_{z=0}, 
\end{equation} 
which tells us that all the $\bra{n}$ are indeed continuous in the normal topology of $\Cz$. 
\end{remark}

The algebraic part of H0 ensures the existence of a basic *-algebra $\azz$.
\begin{definition}
 $\azz$ is defined to be the  
algebra of linear operators acting on $\Cz$ 
generated by the 
operators $z$ and $\zz$. 
Specifically, 
the elements of $\azz$ are 
finite linear combinations of monomials
in $z$ and $\zz$, 
which are by definition all of the 
finite products whose factors 
are either $z$ or $\zz$ in all possible
orders.

Clearly, $\azz$ is a *-algebra, with its *-structure 
being the formal adjoint operation. 
\end{definition}

The existence of the formal adjoint 
implies that 
every operator in $\azz$, 
whose domain $ \Cz $ is dense in 
$\cal{H}$, is closable in $\cal{H}$.  
(See Theorem VIII.1 in \cite{rs1}.)
Moreover, 
\begin{equation}
a\subset \overline{a}\qquad \mathrm{and} \qquad 
\overline{a}^*\supset a^*
\qquad  \mathrm{for~every~}   a\in\azz. 
\end{equation}
In the second formula the $^{*}$ 
in $ \overline{a}^* $
means the Hilbert space 
adjoint operator in $\cal{H}$ of the
densely defined operator $ \overline{a} $. 
Also, note that the closure $ \overline{a} $
is always strictly greater than $ a $.

Depending on $\braket{}{}$ the *-algebra $\azz$ can 
be very simple or arbitrarily complicated. 
The simplest situation is when $z=\zz$, 
in which case the operator of multiplication by $ z $
with dense domain $ \Cz $ is symmetric as a densely 
defined operator acting in $ \cal{H} $. 
This defines the classical context of {\it orthogonal polynomials} over 
$\Bbb{R}$. In this case the scalar product can  
always be represented in the form 
\begin{equation}\label{mu-R}
\braket{\varphi}{\psi}=\int_{\!-\infty}^\infty \overline{\varphi(t)}\psi(t)\, d\mu(t) 
\qquad \mathrm{for~all~} \varphi, \psi \in \Cz 
\end{equation}
with respect to a finite measure $\mu$ on $\Bbb{R}$. Indeed, if $z$ is symmetric, then all the numbers 
$s_{nm}=\braket{z^n}{z^m}=\braketop{0}{z^{n+m}}{0}$ are real. 
This implies that $J$ is isometric, and in particular it extends to an anti-linear isometry of the whole $\cal{H}$. So 
we have an anti-linear isometry $J$ preserving $z$. This symmetry situation implies that $z$ is extendible to a self-adjoint operator in $\cal{H}$, as $J$ exchanges $\ker(\lambda-z^*)$ and $\ker(\bar{\lambda}-z^*)$ 
for every non-real $\lambda \in \Bbb{C}$, and so both defect indices of $z$ are the same. Let us now take any self-adjoint 
extension $\widetilde{z}$ of $z$ and consider its spectral measure $E$. 
This is a projector-valued measure, that is its values $E(\Lambda)$ on Borel subsets $\Lambda$ of $\Bbb{R}$ 
are orthogonal projectors in $\cal{H}$ so that 
$$  \widetilde{z}=\int_\Bbb{R} t\, dE(t).  $$
The desired standard measure $\mu$ reproducing the scalar product $\braket{}{}$ is defined by 
$$ \mu(\Lambda)=\braketop{0}{E(\Lambda)}{0} $$
for Borel subsets $\Lambda$ of $\Bbb{R}$. 

\begin{remark}
Here the topological part of H0 is equivalent to the compactness of
the support of the measure $\mu$. 
\end{remark}

\begin{prop} This classical context of orthogonal polynomials 
is equivalent, modulo 
linear transformations of the complex variable 
$z$, to the commutativity of the *-algebra $\azz$.  
\end{prop}
\begin{proof}  If $\zz$ commutes with $z$, then it belongs to $\Cz$ in the sense that 
$ \zz=\phi(z) $, 
where $\phi$ is a polynomial of degree at least $1$. Here $z$ is viewed in the same way as $\zz$,
namely as an operator in $\Cz$.
But then 
$$
z = z^{**}=\phi(z)^*=(J\phi)(\zz)=(J\phi)[\phi(z)]. 
$$
So the composition of $J\phi$ and $\phi$ is the identity map on $\Bbb{C}$. This is possible only if
\begin{equation}
\label{a-b-conditions}
 \phi(z)=az+b\quad \mathrm{for~} a,b \in \Bbb{C}
\quad \mathrm{satisfying} \quad 
 \vert a\vert=1,\quad a\bar{b}+b=0.
\end{equation} 
Let $ \sqrt{a} $ denote one fixed value for the 
square root of $ a $. 
Using the conditions on $ a,b $
in \eqref{a-b-conditions}
and $ \zz=\phi(z) = a z + b $,   
one then readily calculates that 
$$ \Big(\sqrt{a}z +\frac{1}{2}\frac{b}{\sqrt{a}}\Big)^*=
\sqrt{a}z+\frac{1}{2}\frac{b}{\sqrt{a}}.$$
Therefore the substitution $z\rightsquigarrow \sqrt{a}z+b/(2\sqrt{a})$ does the trick. 
\end{proof}

We can extend \eqref{mu-R} in a sense to our non-commutative context by introducing a canonical 
`integration functional' on $\azz$ defined by
\begin{equation}\label{integral-A}
\int F:=\braketop{0}{F}{0}= 
\braket{\pp_{0}}{F \pp_{0}} = \braket{0}{F} 
\quad \forall F \in  
\azz, 
\end{equation}
where $ \pp_{0} = \ket{0} = 1 $ is the first element in 
the basis \eqref{B-basis-defined}. 
In particular we can write 
\begin{equation}
\braket{\rho}{\pp}=\int\! \rho^*\pp 
\quad \mathrm{for~} \rho, \pp \in \Cz, 
\end{equation}
where on the right side the polynomials 
$ \rho, \pp $
are interpreted 
as elements of $\azz$. 

With the *-algebra $\azz$ we have the very basic algebraic structure for crafting the idea of a quantized Euclidean 
plane. We would like to interpret the operators $z$ and $z^*$ as the 
quantized complex coordinate and its conjugate.  
But the structure still lacks certain 
geometrical and analytical contents, 
which are fundamental for a truly quantum interpretation. 
For instance, we would expect both  
$z$ and $z^*$ to be unbounded operators 
with their spectra  being
equal to the whole plane $\Bbb{C}$. 
And we also want to be able to construct coherent states associated to every complex number. 
And this is not always the 
case within the framework of Harmony Zero. 

Our next harmony property provides these additional
geometrical and analytical dimensions by linking the theory to  Hilbert spaces of entire functions. 

\begin{definition}
We shall say that 
{\em Harmony One} (or simply H1) 
holds (for $\cal{H}$ or for $\braket{}{}$ on $\Cz$) if $\cal{H}$  
can be realized as a Hilbert space of entire functions
on $\Bbb{C}$. 
\end{definition}

\begin{remark} Properties H0 and H1 are logically independent. However,  in diverse interesting examples they 
will happily work together. 
\end{remark}

Harmony One is a symbiosis of two more elementary conditions, which are related to {\it continuity} and {\it injectivity}. 
The algebra of entire functions $\HolC$ possesses its natural normal topology of uniform convergence on 
compact sets. 
All polynomials are entire functions. 
Firstly, we require that the inclusion map $\Cz\rightarrowtail \HolC$ be continuous, where $\Cz$ is considered 
equipped with the $\braket{}{}$-induced topology. 
If so, the inclusion extends to a continuous map 
$\cal{H}\rightsquigarrow \HolC$. 
Secondly, we require that this extended map be injective. 
In such a way all the elements of $\cal{H}$ are viewable as entire functions. 

If, on the other hand, we only have the continuity but not injectivity, 
the situation allows an elegant geometrical 
`renormalization'. 
In this case there is a non-trivial kernel $\cal{K}$ of $\cal{H}\rightsquigarrow \HolC$, which is a closed 
subspace of $\cal{H}$ transversal to $\Cz$. 
If we orthogonally project $\Cz$ onto $\cal{K}^\bot$, the space 
$\Cz$ will be preserved as an everywhere dense projection in $\cal{K}^\bot$. 
This projection can be understood 
as a change of the initial 
scalar product in $\Cz$, so that $\braket{}{}$ is replaced by the scalar product induced from $\cal{K}^\bot$. 
The
restriction of $\cal{H}\rightsquigarrow \HolC$ to  $\cal{K}^\bot$ is by construction injective. 
Hence $\cal{K}^\bot$ 
or equivalently $\Cz$ equipped with the new scalar product, satisfies Harmony One. 

It is instructive to do this construction `backwards' and classify all the scalar products on $\Cz$ for which 
the continuity property holds. 

\begin{prop} There is a canonical 
one-to-one correspondence between the scalar products $\braket{}{}$ on $\Cz$ for which the above 
continuity property holds and triplets $(\cal{K},\braket{}{}^{\!\sim},F)$,  
where $\cal{K}$ is a separable Hilbert space, 
$\braket{}{}^{{\!\sim}}$ a scalar product on $\Cz$ satisfying Harmony One, and $F\colon\Cz\rightarrow\cal{K}$ a completely discontinuous linear map in the sense that  $\mathrm{D}(F^*)=\{0\}$. The space $\cal{H}$ 
is realized as $\cal{H}^{\sim}\oplus \cal{K}$, where $\cal{H}^\sim$ is the Hilbert space completion of $\Cz$ relative to $\braket{}{}^{\!\sim}$. In terms of this identification, the $\braket{}{}$-isometric inclusion 
of $\Cz$ into $\cal{H}$ is $p(z)\longmapsto p(z)\oplus F[p(z)]$.\qed  
\end{prop}

\begin{remark} The map $F$ has the 
	wildest possible behavior for a linear map: its composition with any non-zero continuous functional on $\cal{K}$ is discontinuous. This wildness ensures that the graph of $F$, 
	which realizes the embedding of $\Cz$ into $\cal{H}$,
	 is everywhere dense in $\cal{H}$. 
	
	 To see why this is so, note that
	 $ \Gamma (F^{*}) = 
	 [U (\Gamma (F) )]^{\perp} $, where
	 $ \Gamma (F) $ is the graph of $ F $
	 in  the inner product space 
	 $ \Cz \oplus\cal{K} $ and 
	 $ U : \Cz \oplus\cal{K} \to 
	 \cal{K}\oplus\Cz $ is the inner product
	 preserving map 
	 $ (f,g) \mapsto (-g,f) $.  
	 Similarly, 
	 $\Gamma (F^{*})$ is the graph 
	 of $ F^{*} $.
	 Therefore $ \Gamma (F) $ dense implies that
	 $ \Gamma (F^{*}) = 0 $ and thus 
     $\mathrm{D}(F^*)=\{0\}$. 
     But $\mathrm{D}(F^*)= 
     \{ \kappa \in \cal{K}\,|\,
     \phi \mapsto \braket{\kappa}{F \phi}
     \, \mathrm{is~a~continuous~map~} 
     \Cz \to \mathbb{C} \}$.
     Moreover, $ F $ composed with
     a continuous functional 
     on the Hilbert space 
     $ \cal{K} $ has exactly the form 
     $\phi \mapsto \braket{\kappa}{F \phi} $ 
	 for some $ \kappa \in \cal{K} $. 
	 So $\mathrm{D}(F^*)=\{0\}$ implies that
	 this composition is continuous only 
	 for the zero functional. 
	 
	 Interestingly, such an extreme discontinuity naturally emerges in the study of an important continuity property. 
	 If this continuity is granted, then the
	  injectivity property holds for $\braket{}{}$ if
	  and only if $\cal{K}=\{0\}$ which is itself 
	  equivalent to $\braket{}{}=\braket{}{}^{\!\sim}$. 
\end{remark}

\begin{remark} 
On the other hand under some 
appropriate additional symmetry 
conditions on the scalar product, the injectivity of the extended map 
$\cal{H}\rightsquigarrow \HolC$ automatically holds. As we shall see, this includes our primary scenario when the monomials $z^n$ are mutually orthogonal. 
\end{remark}

If Harmony One holds, then for all $ z \in \Bbb{C} $
the linear functional 
$ l_{z} : \cal{H} \rightarrow  \Bbb{C}$
defined for all $ f \in \cal{H} $
by $ l_{z} (f) := f(z) $ (called {\it evaluation 
at $ z $}) is continuous with respect to the norm 
topology of $ \cal{H} $. 
To see that this is so, note that 
uniform convergence on compact sets implies 
convergence on the compact, singleton set $ \{ z \} $, 
that is, point-wise convergence. 
And that tells us that 
$ l_{z} : \HolC \to \Bbb{C} $,
defined by the same formula as above, 
is continuous in the 
normal topology. 
But the inclusion map  
$\cal{H}\rightsquigarrow \HolC$ is continuous by 
Harmony One. 
So the composition of these two maps is continuous, 
as claimed.
This in turn implies that $ \cal{H} $ is 
a {\em reproducing kernel Hilbert space}.
So 
in accordance with the general theory of Hilbert spaces of entire functions 
(for example, see \cite{debranges}) 
there exists a unique 
{\em reproducing kernel}
$K\colon\Bbb{C}\times\Bbb{C}\rightarrow \Bbb{C}$ of $\cal{H}$ given by the series 
\begin{equation}\label{K}
K(\bar{w},z)=\sum_{n=0}^\infty \overline{\pp_n(w)}\pp_n(z),  
\end{equation}
which converges absolutely and uniformly on compact sets of $\Bbb{C}\times\Bbb{C}$. 
The following are the characteristic
properties of $ K $: 
\begin{itemize}
	\item 
	For every $ z \in \Bbb{C} $ and $ f \in \cal{H} $ 
	we have the {\em Reproducing Property:} 
	$$
	 f (z) = \int_{\Bbb{C}} K (\bar{w},z) f (w) d w
	$$
	\item 
	For each $ w \in \Bbb{C} $ the function 
	$ \Bbb{C} \ni z \mapsto K (\bar{w},z) $ 
	is an element in $ \cal{H} $. 
\end{itemize}

It follows in this case from \eqref{n-kernel}
that  the sequence of 
inverse matrices $S_n^{-1}$, viewed naturally in $\mathrm{M}_\infty(\Bbb{C})$, converges entry-by-entry 
as $ n \to \infty $
to an infinite hermitian matrix $\sinfty^-$ so that  
\begin{equation}\label{wz-kernel}
K(\bar{w},z)=\sum_{i,j=0}^\infty z^i[S_{\infty}^-]_{ij}\bar{w}^j, 
\end{equation}
a series which is normally convergent on $\Bbb{C}\times\Bbb{C}$. 

\begin{remark} An additional explanation for this important 
inversion formula is, perhaps, in order here. We know that 
the partial sums, the left hand side of \eqref{n-kernel}, converge normally on $\Bbb{C}\times\Bbb{C}$ to the 
reproducing kernel $K(\bar{w},z)$. On the other hand, the reproducing kernel is expandable into a double 
power series in $\bar{w}$ and $z$, normally convergent on the whole $\Bbb{C}\times\Bbb{C}$. In particular, 
the coefficients of the partial sums, the right-hand side of \eqref{n-kernel}, converge as $n\to\infty$ to 
the coefficients of the expansion of the reproducing kernel, and \eqref{wz-kernel} indeed holds. 
\end{remark}

For each $w\in\Bbb{C}$ its 
{\em point wave function} is defined 
for all $ z \in \Bbb{C} $ by 
$$
w(z):=K(\bar{w},z) \quad 
\mathrm{or~more~simply~by} \quad [w]:= K(\bar{w},\cdot). 
$$ 
We have that 
$[w] = w \in \cal{H}$ by the 
second characteristic property of $ K $. 
(That both $w\in\Bbb{C}$ and  $  w \in \cal{H} $ 
is an abuse of notation which will be be resolved 
by context.) 
Also by the reproducing property of $ K $ we obtain 
\begin{equation}\label{w(z)}
\braket{[w]}{\psi} = 
\braket{w(z)}{\psi(z)}=\psi(w) 
\end{equation}
for every $\psi=\psi(z)\in\cal{H}$ 
and $ w \in \Bbb{C} $. 
In particular, for the point wave functions 
associated to $ v, w \in \Bbb{C} $
we obtain
\begin{equation}\label{w-norm}
\braket{[v]}{[w]} = 
\braket{v(z)}{w(z)}=K(\bar{w},v)
\quad \mathrm{and}\quad 
\| \,[w]\, \|^2 = 
\|w(z)\|^2=K(\bar{w},w). 
\end{equation}

We claim that a point wave function $ [w] $ can not 
be identically equal to zero. 
In fact, by \eqref{w(z)} the set of 
{\em dormant points} defined as 
$  \{ w \in \Bbb{C} ~|~ [w] = 0 \} $ is equal to 
$$
 \{ w \in \Bbb{C} ~|~ \psi (w) = 0  \quad \forall 
 \psi \in \cal{H} \}. 
$$
But $ \cal{H} $ contains all the elements in $ \Cz $, 
that is, all polynomials. 
And there is no complex number that is the common 
zero of all polynomials. 
So, in this setting there are no dormant points
and consequently 
$ [w] \ne 0$ for all $ w \in \Bbb{C} $. 
However, in other more general reproducing kernel 
Hilbert spaces there are dormant points. 

The following is a standard result, which 
goes much beyond $ [w] \ne 0$. 
\begin{lemma}
\label{LI-lemma}
The set of vectors $\{ ~ [w] ~|~ w \in  \Bbb{C} \} $ 
in $ \cal{H} $ is linearly independent.
\end{lemma}

\begin{proof}
A set of vectors is linearly 
independent if and only if all of its 
finite subsets are linearly 
independent. 
So let $ w_{1}, \dots , w_{n} $ be  $ n $
distinct points in $ \Bbb{C} $. 
We have to show that the vectors 
$  [w_{1}], \dots , [w_{n}]  $ are linearly 
independent. 
So suppose that 
$ \sum_{j=1}^{n} \lambda_{j} [w_{j}] = 0$ for some
$ \lambda_{j} \in \Bbb{C} $. 
We have to prove that 
$ \lambda_{j} = 0 $ for all $ j $. 
Now for all $ \psi \in \cal{H} $ we have
\begin{equation}
\label{prove-LI}
0 = \braket{\sum_{j=1}^{n} \lambda_{j} [w_{j}]}{\psi} = 
\sum_{j=1}^{n} \overline{\lambda_{j}}
\braket{[w_{j}]}{\psi} =
\sum_{j=1}^{n} \overline{\lambda_{j}} \psi(w_{j}). 
\end{equation}
Since the points 
$ w_{1}, \dots , w_{n} $ are distinct, 
for each $ 1 \le i \le n $ there exists a polynomial 
$ p_{i}(z) \in \Cz \subset \cal{H} $ such that $ p_{i} (w_{j}) = \delta_{ij} $, the Kronecker delta, 
for each $ 1 \le j \le n $. 
(These are the basis Lagrange polynomials 
for the points $ w_{1}, \dots , w_{n} $.) 
Taking $ \psi = p_{i} $ in \eqref{prove-LI} 
shows that $ \lambda_{i} = 0 $ 
for every $ 1 \le i \le n $
\end{proof} 

Let us consider the linear span (finite linear combinations) of all the point wave functions:
\begin{equation}
\label{cal-L-defined}
\cal{L}:=\Bigl\{\sideset{}{^*}\sum_{w\in\Bbb{C}} c_w w(z)\Bigr\}.
\end{equation}

This is a fundamental object directly emerging from the 
reproducing kernel. 
From \eqref{w(z)} we conclude that there is no non-zero 
vector in $ \cal{H} $ that is orthogonal to all of the
point wave functions. 
Consequently, the vector subspace $\cal{L}$ is dense in 
$\cal{H}$ in the norm topology and hence also 
dense in 
$\cal{H}$ in the normal topology. 
Moreover, because of \eqref{w-norm} both the 
scalar product and the norm restricted to $\cal{L}$ 
are completely determined by the reproducing kernel.  

The reproducing kernel formula is a key to expressing Harmony One in terms of the canonical orthonormal basis.
 
\begin{prop}\label{harmony-one} Harmony One holds if and only if the series 
\begin{equation}\label{z-ppn-series}
\sum_{n=0}^\infty |\pp_n(z)|^2
\end{equation}
is normally convergent on $\Bbb{C}$,  
and in addition the orthonormal 
polynomials $\pp_n(z)$ are $\infty$-linearly independent in the sense that all infinite linear combinations
\begin{equation}\label{infty-linear}
\sum_{n=0}^\infty c_n\pp_n(z) \quad 
\mathrm{with} \quad 
\sum_{n=0}^\infty |c_n|^2<+\infty, 
\end{equation}
understood as entire functions, uniquely determine their coefficients $c_n$. 
\end{prop}

\begin{proof} The $\Rightarrow$ part is clear as in this case 
\eqref{z-ppn-series} normally converges to $K(\bar{z},z)$, and 
the functions \eqref{infty-linear} are faithful representations of the vectors of $\cal{H}$. 

The $\Leftarrow$ part 
is a little subtler. 
If \eqref{z-ppn-series} is normally convergent on $\Bbb{C}$, 
then we can {\it define} the 
function $K(\bar{w},z)$ by the series 
\eqref{K} for which it is easy to see that 
it converges normally 
on $\Bbb{C}\times\Bbb{C}$ in this case.
Moreover, this function satisfies by construction the primary 
matrix positivity condition (as discussed in Appendix~B). 
So it is the reproducing kernel function 
for a Hilbert space $\cal{H}^\sim$ of 
entire functions. 
Let $\cal{L}^\sim$ be its dense linear subspace of 
finite linear combinations of 
point wave functions. 
The formula 
\begin{equation*}
\cal{L}^\sim\ni w(z)\rightsquigarrow \sum_{n=0}^\infty \overline{\pp_n(w)}\pp_n\in\cal{H}
\end{equation*} 
for all $ w \in \Bbb{C} $ 
defines an isometric embedding of $\cal{L}^\sim$ into $\cal{H}\leftrightarrow \ell^2(\Bbb{N})$. 
This isometric 
embedding extends to an isometric embedding of $\cal{H}^\sim$ into $\cal{H}$. 
Using this embedding, 
let us calculate the orthocomplement of 
$\cal{H}^\sim$ in $\cal{H}$, namely 
$$\cal{H}^{\sim\bot}=\cal{L}^{\sim\bot}=\bigcap_{w\in\Bbb{C}}[w(z)]^\bot. $$
Taking into account that 
$$\braket{w(z)}{\sum_{n=0}^\infty c_n\pp_n}=\sum_{n=0}^\infty c_n\pp_n(w) $$
we see that this orthocomplement consists precisely of infinite square summable decompositions of $0$ into 
$\infty$-linear combinations of the 
functions $\pp_n(z)$ as in \eqref{infty-linear}. 
The orthocomplement 
will be trivial (equivalently $\cal{H}=\cal{H}^\sim$) 
if and only if 
\eqref{infty-linear} faithfully represents the vectors of $\cal{H}\leftrightarrow\ell^2(\Bbb{N})$. 
\end{proof}

\begin{remark} In fact the normal convergence of the series \eqref{z-ppn-series} is equivalent to the continuity 
part of Harmony One, while the faithfulness of \eqref{infty-linear} is equivalent to the injectivity part. 
\end{remark}

The point wave functions give us corresponding 
coherent states, since they can be normalized.
By using some more Dirac notation,  
there is a natural
embedding 
\begin{equation}\label{w-w(z)}
\Bbb{C} \ni 
 w\rightsquigarrow \frac{\diada{w(z)}{w(z)}}{\braket{w(z)}{w(z)}}
 = \frac{\diada{[w]}{[w]}}{|| \,[w]\, ||^{2}}
 \in \mathrm{CP}(\cal{H})
\end{equation}
into the complex projective space $\mathrm{CP}(\cal{H})$, which is naturally 
identified as the space of pure states of $ \cal{H} $. 
In the context given by $\cal{H}$
the {\em coherent states} 
$ \diada{[w]}{[w]} / || \, [w] \, ||^{2} $
can be interpreted as the quantum counterparts of classical points in the plane $\Bbb{C}$.   
In the literature
the corresponding unit vectors 
$  [w]  / || \, [w] \, ||  \in \cal{H} $ 
are also called coherent states.

\begin{prop}The induced metric on $\Bbb{C}$ from $\mathrm{CP}(\cal{H})$ via the above map is given by
\begin{equation}\label{K-induced}
\de s^2=\de \bar{w}\,\de w\Bigl\{\frac{\partial^2}{\partial \bar{w}\,\partial w} \,\log K(\bar{w},w)\Bigr\}. 
\end{equation}
\end{prop}

\begin{proof} The complex projective space,
viewed as the set of rank 1 projectors as above,  
is a subset of the space of Hilbert-Schmidt operators
acting in $\cal{H}$, which 
itself is a Hilbert space relative to the scalar product defined by 
\begin{equation}\label{L2-H}
\braket{X}{Y}:=\frac{1}{2}\mathrm{Tr}(X^*Y).
\end{equation} 
By definition, 
the metric on $\mathrm{CP}(\cal{H})$ is the corresponding 
induced metric. 
And in the coordinates $w,\bar{w}$  we obtain
\begin{multline*}
\frac{1}{2}\mathrm{Tr}\Bigl\{\Bigl(\frac{\ket{w+\de w}\bra{w+\de w}}{
\langle w+\de w\vert w+\de w\rangle}-\frac{\ket{w}\bra{w}}{
\langle w\vert w\rangle}\Bigr)^2\Bigr\}\sim \\ 
\de s^2=\Bigl\{K(\bar{w},w)\frac{\partial^2}{\partial \bar{w}\,\partial w}K(\bar{w},w)
-\frac{\partial}{\partial \bar{w}} K(\bar{w},w)\, \frac{\partial}{\partial w} K(\bar{w},w)\Bigr\}\frac{\de \bar{w}\:\de w}{
K(\bar{w},w)^2}
\end{multline*}
where we have disregarded the higher order terms in $\de w$ and $\de\bar{w}$. The formula 
\eqref{K-induced} follows elementarily. 
\end{proof}

\begin{remark} The above formula is valid for all 
reproducing kernel Hilbert spaces of analytic
functions over arbitrary domains $\Omega$ in 
$\Bbb{C}$ as long as the point vectors
$w(z)$ are all non-zero. In other words, no point 
$w\in\Omega$ is a common zero 
for all the elements of $\cal{H}$.
This is a wise thing to assume always, since in the
contrary case we can simply remove 
the discrete set of dormant points from $\Omega$.  
However, as we already showed, there are no dormant 
points in the setting of this paper. 
\end{remark}

\begin{remark} The factor $1/2$ in the definition 
\eqref{L2-H} is justified by the simple 
resulting expression for 
the induced metric in $\Bbb{C}$. 
Another justification is that in this case 
for unit vectors $ \varphi, \psi $ 
the square of the 
distance between the one-dimensional projectors $\diada{\psi}{\psi}$ and $\diada{\varphi}{\varphi}$ has a direct physical meaning: the value is 
$1-|\braket{\varphi}{\psi}|^2$, which is the quantum 
mechanical probability of {\it not} transitioning from 
one of the states 
$ \varphi, \psi $ to the other. 
Intuitively, this probability of non-transition 
(or more accurately, its square root) is 
measuring how far apart the states determined by  
$ \varphi, \psi $ are from each other.  
\end{remark}

\begin{prop}
\label{O-C-extension-thm}
Assume that Harmony One holds. 
Then every transformation $\tau\in\Opiz$ naturally 
induces a closed linear operator, denoted by 
$\widehat{\tau}$, acting in $\cal{H}$ on the 
not necessarily dense domain defined by 
\begin{equation}\label{D-tau}
\mathrm{D}(\widehat{\tau}):=
\Bigl\{\psi(z)\in\cal{H}
\Bigm\vert \tau[\psi(z)]\in\cal{H}\Bigr\}
\end{equation} 
on which $\widehat{\tau}$ acts by 
$\widehat{\tau}\colon \psi(z)\mapsto \tau[\psi(z)]$. 
 \end{prop}

\begin{remark}
The operator $\widehat{\tau}$ is 
the maximal restriction of $\tau$ within $\cal{H}$. 
Its domain can be trivial, consisting 
of $\{0\}$ only. The most interesting 
situations occur when its domain is dense 
in $\cal{H}$ so that the adjoint operator 
$\widehat{\tau}^*$ is defined
in $\cal{H}$.
Then, because of the density of $\cal{H}$ 
in $\HolC$ in the normal topology, $\tau$ is completely 
fixed by $\widehat{\tau}$. This includes the maximal 
compatibility $\tau(\cal{H})\subseteq\cal{H}$, 
in other words $\mathrm{D}(\widehat{\tau})=\cal{H}$, which 
by the closed graph theorem, implies the boundedness of $\widehat{\tau}$.  
\end{remark}

\begin{proof} 
The closeness follows from the continuity
of $\tau$ in the normal topology. 
The 
graph $\mathrm{G}(\tau)$ of $ \tau $ is closed in $\HolC\oplus\HolC$. 
The natural inclusion map 
$$ \cal{H}\oplus\cal{H}\rightarrowtail \HolC\oplus\HolC $$
is continuous. 
So the inverse image of $\mathrm{G}(\tau)$, 
which is precisely the graph of $\widehat{\tau}$, 
is closed in $\cal{H}\oplus\cal{H}$. 
\end{proof}

In examples many important operators $\tau$ 
will come from $\Opz$. 
In particular, we see that such
 operators include 
all the multiplication operators 
by polynomials in $z$
(i.e., by the elements in $\Cz$, 
which includes multiplication by $z$),  
the derivative operator $\partial/\partial z$ and 
all linear differential operators 
with polynomial coefficients in $ z $
and also the substitution operators 
by polynomials $\pp(z)$. 

In what follows we use the notation of the previous 
proposition to define the symbol 
$Z:= \widehat{z}$, that is, the closed 
operator of multiplication 
by the coordinate $z$ acting in a domain of 
$ \cal{H} $ defined by 
\begin{equation}\label{D-Z}
\mathrm{D}(Z):=\Bigl\{\psi(z)\in\cal{H}\Bigm\vert z\psi(z)\in\cal{H}\Bigr\}. 
\end{equation}
So $Z$ acts by $Z\colon\psi(z)\mapsto z\psi(z)$. 
The operator $Z$ is intrinsically related to the point 
wave functions $w(z)$ as we shall now explain. 

Let us observe for all $ w \in \Bbb{C} $ we have that 
\begin{equation}\label{w-orthogonal}
 w(z)^\bot=\Bigl\{\psi(z)\in\cal{H}\Bigm | \psi(w)=0\Bigr\}=\overline{(z-w)\Cz},
 \end{equation} 
 where the superscript line on the right denotes the 
 closure in the norm topology of $ \cal{H} $. 
Indeed, the first equality follows from the property
\eqref{w(z)} 
 of the point wave functions, and the second equality 
means that the functions from $\cal{H}$ vanishing at $w$ can always be 
arbitrarily well approximated in $\cal{H}$ by polynomials with the same property.

\begin{prop} 
For every $w\in\Bbb{C}$ its associated point 
function $[w]=w(z)$ belongs to the 
domain of the adjoint operator $Z^*$. 
In particular, $ Z^* $ is densely defined. 
Moreover, $ Z^* [w]=\bar{w}\,[w]$ and  
\begin{equation}\label{One-D}
\ker(\bar{w}-Z^*)=\Bbb{C}[w]. 
\end{equation}
The operators $Z$ and $Z^*$ are unbounded.  
Every complex number $ w $ is an eigenvalue of 
the annihilation operator
$ Z^* $ of
multiplicity one and with associated 
eigenvector being the normalized point 
wave function 
(equivalently, coherent state) 
$ [ \bar{w} ]/ || \,[\bar{w}]\,  || $. 
\end{prop}

\begin{proof}
Take $w\in\Bbb{C}$. 
From the principal kernel equation \eqref{w(z)} 
applied twice we find 
$$ \braket{[w]}{Z\psi}=w\psi(w)=
w \braket{[w]}{\psi}
= \braket{\bar{w}[w]}{\psi}
$$
for every $\psi\in\mathrm{D}(Z)$
and every $w\in\Bbb{C}$. 
Therefore $[w]$ belongs to the domain of $Z^*$ and 
$ Z^* [w]=\bar{w}\,[w].$
In particular $\mathrm{D}(Z^*)\supseteq \cal{L}$ (cp. \eqref{cal-L-defined}), and so $ Z^* $ 
is densely defined in $ \cal{H} $. 
Since $ [w] \ne 0 $, 
it follows that 
$ \bar{w} $ is an eigenvalue of $ Z^*$ for 
every $ w \in \Bbb{C} $.
So the spectrum of $Z^*$ is $ \Bbb{C} $. 
So $Z^*$ is an unbounded operator, which in 
turn implies that $ Z $ is also an unbounded operator. 
In particular $\mathrm{D}(Z)\neq \cal{H}$ 
and $\mathrm{D}(Z^*)\neq \cal{H}$.  

Having shown that every 
point in $\Bbb{C}$ is an eigenvalue 
of $ Z^* $, 
we now identify its multiplicity.  
Using a standard identity for the adjoint operator
we have 
$$ 
\ker(\bar{w}-Z^*)^\bot=\overline{\mathrm{im}(w-Z)} =
[w]^\bot.
$$
where the second identity follows from \eqref{w-orthogonal}. 
We conclude that $\ker(\bar{w}-Z^*) = \Bbb{C} [w]$,  
which shows \eqref{One-D}. 
Now $ \Bbb{C} [w] $ is 
a one dimensional subspace because $ [w] \ne 0 $. 
And this shows that each eigenvalue of $ Z^* $ 
has multiplicity one. 
\end{proof}

\begin{remark}
There is a certain complementarity between the two basic spaces 
$\Cz$ and $\cal{L}$: The space $\Cz$ is $Z$-invariant and the space $\cal{L}$ is $Z^*$-invariant. 
In a variety of interesting 
examples we encounter one of the following two special configurations:  
\begin{align}
\Cz&\cap\cal{L}=\Bbb{C}\label{Cz-L-C}\\
\Cz&\subseteq \cal{L}.\label{Cz-in-L}
\end{align}
In general, the finite dimensional 
spaces $\Czf{n}$ can only contain a finite number 
of the linearly independent 
point wave functions $w(z)$. 
(Cp. Lemma~\ref{LI-lemma}.)
Hence, at most countably 
many of them will be in $\Cz$. 

Since $ Z^* $ is densely defined and $ Z $ is closed,
$ Z^{**} = \bar{Z} = Z $. 
So $ Z $ and $ Z^* $ are adjoints of each other. 
\end{remark}

\begin{prop}\label{Z*-closure} 
The subspace $ \cal{L} $ of $ \cal{H} $
is a core for the operator $Z^*$. 
In other words
$Z^*$ is the closure of its restriction $Z^*\vert\cal{L}\colon\cal{L}\rightarrow\cal{L}$, 
that is $ \overline{Z^*\vert\cal{L}} = Z^* $. 
\end{prop}

\begin{proof} 
	Let us compute the orthocomplement of the graph $\mathrm{G}(Z^*\vert\cal{L})$ in the graph 
$\mathrm{G}(Z^*)$. 
It consists of the 
pairs $(\psi, Z^*\psi)$ for some
$\psi\in\mathrm{D}(Z^*)$ which are 
orthogonal to all the pairs $([w],\bar{w}[w])$ where
$w\in\Bbb{C}$. In other words $\braket{[w]}{\psi}+\braket{\bar{w}[w]}{Z^*\psi}=0$. But this means $\psi(w)+w(Z^*\psi)(w)=0$.
We now claim under this hypothesis on $ \psi $ that 
$ Z^* \psi \in \mathrm{D}(Z)$. 
This is because  
$ \Bbb{C} \ni w \rightsquigarrow w Z^* \psi (w) = - \psi (w) $ 
is the function $ -\psi \in \cal{H} $. 
It then follows from the definition of $ Z $ that 
$ (Z Z^* \psi) (w) = w Z^* \psi (w)$ holds for 
all $ w \in \Bbb{C} $ for this particular 
vector $ \psi $. 

Next, this says that $\psi(w)+(ZZ^*\psi)(w)=0$ for 
all $ w \in \Bbb{C} $, 
and therefore $ \psi $ satisfies 
$\psi+ZZ^*\psi=0$. 
By functional analysis
$1+ZZ^*$ is invertible with 
its inverse defined on the whole 
$\cal{H}$ and bounded. 
It follows that $\psi=0$, and so the
orthocomplement is the zero subspace. 
Thus $\mathrm{G}(Z^*\vert\cal{L})$ is dense in 
$\mathrm{G}(Z^*)$. 
\end{proof}

\begin{remark} Another way to see this is by considering the adjoint of $Z^*\vert\cal{L}$. For $\psi$ in the domain of 
this adjoint, 
that is $ \psi \in \mathrm{D} ((Z^*\vert\cal{L})^*) $,  
we have that the function 
$$
\Bbb{C} \ni w \rightsquigarrow w\psi(w)=\braket{Z^*[w]}{\psi}=\braket{[w]}{(Z^*\vert{\cal{L}})^*\psi}=((Z^*\vert{\cal{L}})^*\psi)(w)$$ 
is in $ \cal{H} $, 
which shows us that in fact 
$ \psi \in \mathrm{D} (Z)$ 
and 
$Z \supset (Z^*\vert{\cal{L}})^*$.
And therefore  
$Z^* \subset (Z^*\vert{\cal{L}})^{**} = 
\overline{Z^*\vert{\cal{L}}}$. 
The opposite inclusion is trivial, and so 
$ Z^* = \overline{Z^*\vert{\cal{L}}} $.   
Interestingly, a similar statement does not in general hold for the operator $Z$. 
In general $Z$ will be strictly greater than the closure of its restriction on $\Cz$. 

The above proposition extends to all the powers of $Z^*\vert{L}$, as 
$$ (Z^*\vert{\cal{L}})^{n*}=Z_n,$$ 
where $Z_n\colon\psi(z)\rightarrow z^n\psi(z)$ is the operator associated to the multiplication by $z^n$. In general, this operator will be 
a non-trivial extension of $Z^n$ although in many important examples $Z_n=Z^n$.  
We refer to \cite{bs} for an elegant and self-contained geometrically oriented exposition of the theory of unbounded operators in a Hilbert space. 
\end{remark}

As we have already mentioned, the polynomials 
$ \phi \in \Cz $ 
can be viewed 
as unital homomorphisms of $\Cz$ into itself 
given for $ p \in \Cz $  by 
the `substitution map'  
$ p \mapsto p \circ \phi $, 
where $ \circ $ denotes composition of functions. 
Since every such a 
map is continuous in the normal topology,
by Proposition \ref{O-C-extension-thm}
we can associate to it a closed operator 
$\widehat{\pp}$, which is obviously densely defined,  
by $ \psi \mapsto \psi\circ\pp $ for all 
$ \psi \in \cal{H} $ satisfying 
$ \psi\circ\pp \in \cal{H} $. 
 
\begin{prop} 
Suppose that $ \pp \in \Cz $. 
Then the domain of the adjoint operator $\widehat{\pp}^*$ contains the space $\cal{L}$. 
We have for every $w\in\Bbb{C}$ that 
\begin{equation}\label{pp-w(z)}
\widehat{\pp}^*([w])=[\pp(w)]. 
\end{equation}  
In particular, the subspace $\cal{L}$ is $\widehat{\pp}^*$-invariant. 
\end{prop}

\begin{proof} 
Applying the kernel equation \eqref{w(z)} 
we find for every 
$\psi\in\mathrm{D}(\widehat{\pp})$ that
$$ \braket{[w]}{\widehat{\pp}\psi}=
\braket{[w]}{\psi\circ\pp}=
\psi(\pp(w))
=\braket{[\pp(w)]}{\psi}. $$
Thus $[w]\in\mathrm{D}(\widehat{\pp}^*)$ and \eqref{pp-w(z)} holds. 
\end{proof}

\begin{remark} The above argument is extendible without essential 
change from polynomials 
to entire functions, understood as unital homomorphisms of $\HolC$ into itself, continuous in the 
normal topology. The only subtlety is that we have to explicitly postulate the density of 
$\mathrm{D}(\widehat{\pp})$ in $\cal{H}$. 
\end{remark}

The point wave functions are special case of an important class of vectors in $\cal{H}$, which correspond to the linear functionals on $\cal{H}$ continuous in the normal topology. 

\begin{definition} A vector $\xi\in\cal{H}$ is called {\it normal} if its dual vector, the linear functional $\bra{\xi}\colon\psi\mapsto\braket{\xi}{\psi}$, is continuous in the normal topology of $\cal{H}$.  
\end{definition}
 
We shall write $\W$ for the set of all normal vectors of $\cal{H}$. 
Clearly, $\W$ is a linear subspace of  $\cal{H}$ and 
$\cal{L}\subseteq\W$. 
So $\W$ is dense in $\cal{H}$.  

\begin{remark}
In the framework of Harmony One, the topological part of Harmony Zero can be rephrased 
as the inclusion $\Cz\subset \W$, that is, 
the polynomials are normal vectors. 
\end{remark}

\begin{lemma} If polynomials are normal, 
then the infinite matrix $S_\infty$ is invertible and 
\begin{equation}
S_\infty^-=S_\infty^{-1}. 
\end{equation}
In particular, the extended scalar product $\braket{}{}\colon\Cz\times\HolC\rightarrow\Bbb{C}$ is non-degenerate. 
\end{lemma}

\begin{proof} From the expansion 
\eqref{wz-kernel} we find for all $ k \in \Bbb{N} $ 
and $ w \in \Bbb{C} $ that  
$$
\braket{z^k}{w(z)}=
\braket{z^k}{\sum_{i,j\geq 0}z^i[S_{\infty}^-]_{ij}\bar{w}^j}=
\sum_{i,j\geq0}
[S_\infty]_{ki}[S_{\infty}^-]_{ij}\bar{w}^j)
=\bar{w}^k, 
$$
where the second equality follows from the 
hypothesis that polynomials
are normal. 
And thus $S_{\infty}^-S_{\infty}=1_\infty$. Because of the hermicity of $S_{\infty}$ and $S_\infty^-$ we 
also have $S_{\infty}S_{\infty}^-=1_\infty$. 
\end{proof}

\begin{prop}\label{W-inclusive} 
Let a transformation $\tau\in\Opiz$ be such that 
$\overline{\mathrm{D}(\widehat{\tau})}=\cal{H}$, 
where $ \widehat{\tau} $ is defined in 
Proposition \ref{O-C-extension-thm}. 
Then $\W
\subseteq \mathrm{D}(\widehat{\tau}^*)$ and,
in particular, $ \widehat{\tau}^* $ is densely 
defined. 
Moreover,  $\widehat{\tau}^*(\W)\subseteq\W$.  
\end{prop}

\begin{proof} For every $\xi\in\W$ the map  $\mathrm{D}(\widehat{\tau})\ni\psi\longmapsto 
\braket{\xi}{\widehat{\tau}\psi}\in\Bbb{C}$, being the composition of the
inclusion $\mathrm{D}(\widehat{\tau})$ 
into $\HolC$, the normally continuous tranformation $\tau$ and normally continuous 
scalar product with $\xi$, is normally continuous and in particular $\braket{}{}$-continuous. 
Therefore such vectors $\xi$ are in the domain of the adjoint 
operator for $\widehat{\tau}$. The same initial map now viewed as 
$\mathrm{D}(\widehat{\tau})\ni\psi\longmapsto 
\braket{\widehat{\tau}^*\xi}{\psi}\in\Bbb{C}$ for being normally continuous, and  
since $\mathrm{D}(\widehat{\tau})$ is dense in $\cal{H}$, 
extends by continuity to a normally continuous functional
on the whole $\cal{H}$, given by the same formula. 
It follows that $\widehat{\tau}^*\xi\in\W$. 
\end{proof}

In particular we see that $\W\subseteq \mathrm{D}(Z^*)$. 
So if Topological H0 holds, then $\Cz$ is also 
in the domain of $Z^*$. 
And if in addition Algebraic H0 holds (so that we have full H0), 
then $Z^*\supset z^*$, 
where $ z $ denotes the 
operator of multiplication by $ z $ acting on $ \Cz $

\begin{prop}\label{W-invariant} 
If a transformation $\tau\in\Opiz$ satisfies 
$\overline{\mathrm{D}(\widehat{\tau})}=\cal{H}$
(as in the previous proposition) and  
$\widehat{\tau}^*$ is also continuous in 
the normal topology of $\cal{H}$, then 
$\W\subseteq\mathrm{D}(\widehat{\tau})$ and $\widehat{\tau}(\W)\subseteq\W$. 
\end{prop}

\begin{proof} 
Since $\widehat{\tau}^*$ is normally continuous,
it extends 
by continuity to a transformation $\varrho\in \Opiz$. Clearly $\widehat{\varrho}\supseteq\widehat{\tau}^*$ and 
so $\widehat{\tau}=\widehat{\tau}^{**}\supseteq \widehat{\varrho}^*$. 
In particular, 
$ \mathrm{D}( \widehat{\varrho}^* ) \subseteq
\mathrm{D}( \widehat{\tau} )$. 
And by the previous proposition 
$ \W \subseteq  \mathrm{D}( \widehat{\varrho}^* )$.
Combining these two inclusions gives 
$\W\subseteq\mathrm{D}(\widehat{\tau})$. 
Since $\W$ is $\widehat{\varrho}^*$ invariant
(again by the previous proposition), 
$ \W $ is also invariant for its extensions such as $\widehat{\tau}$. 
\end{proof}

\subsection*{A Minimal Effective Operational Setting}

Our next harmony condition quite naturally emerges when looking for the most effective setting 
for constructions involving the fundamental unbounded operators $Z$ and $Z^*$. 
 
\begin{definition}
We shall say that {\it Harmony Two} (or simply H2) 
holds if the operator $Z^*$ is continuous 
in the normal topology of $\cal{H}$.  
\end{definition}
 
Before going further, some important observations are due. It is clear that H2 requires H1 as it builds on top of 
it. However, H2 is logically independent of H0.  

The operator $Z$ comes from a map in $\Opz$, 
and hence it is in $\Opiz$. 
However, 
in general, the adjoint operator $Z^*$ will not be continuous in the normal topology. 
Therefore, there is an inherent contextual asymmetry 
between the two operators $ Z $ and $ Z^* $. 
This symmetry is imposed by this new harmony property. 
In a sense, Harmony Two 
is a `grown up' version of Algebraic H0, 
which provides the simplest common setting for 
a strictly polynomial version 
of symmetry between the operators $Z$ and $Z^*$.   
 
For $Z^*$ to be normally continuous 
it is sufficient for the restriction $Z^*\vert\cal{L}$ 
to be normally continuous, 
since $ \cal{L} $ is dense in the normal topology 
in $ \cal{H} $ and so also in $ \mathrm{D} (Z^*) $. 
If $Z^*$ is normally continuous, 
then it uniquely extends by continuity 
to $\HolC$, and so it is interpretable as a transformation of $\Opiz$.
We shall use the inverted symbol $\zzinfty$ 
to denote this transformation  
$\zzinfty\colon\HolC\rightarrow\HolC$.   

Therefore, if H2 holds, there is a naturally emerging subalgebra $\azz$ of $\Opiz$ defined 
as the set of all   
(in general, non-commutating) polynomials in $z$ and $\zzinfty$. 

\begin{prop} 
\label{nice-W-properties}	
In the framework of Harmony Two 
the space $\W$ is both $Z$ and $Z^*$-invariant. 
Hence, it is 
invariant for all the operators in $\azz$. 
The algebra $\azz$ is faithfully realized in $\W$ and
is equipped with a 
natural *-structure induced by the formal adjoint operation on $\W$. 
\end{prop}

\begin{proof} 
A direct application of Proposition~\ref{W-invariant} 
implies that $\W$ is 
contained in the domains of $Z$ and of $Z^*$, 
and it is a common invariant subspace 
of these operators. 
When restricted to $\W$ the mutual 
adjointness of $Z$ and $Z^*$ becomes the formal adjointness of their restrictions. The entire algebra $\azz$ 
is representable by restrictions in $\W$. Since $\W$ is norm dense in $\cal{H}$ it is also normally 
dense in $\HolC$, which means that the representation of $\azz$ in $\W$ is faithful.   
\end{proof}

\begin{remark} The same argument extends to all operators $\tau\in\Opiz$ having  
densely defined $\widehat{\tau}$ whose adjoint $\widehat{\tau}^*$ is normally continuous. 
\end{remark}

The algebra $\azz$ is highly non-commutative. 
One manifestation of this would be the absence of 
characters 
($\Leftrightarrow$ one-dimensional *-representations) 
of $\azz$. 
Quantum mechanically this 
means that $\azz$, viewed as an algebra of observables, 
does not admit dispersion-free states.  
And geometrically, the interpretation is 
that the `underlying' quantum Euclidean plane, 
whose `functions' are the elements of $\azz$,  
possesses no classical points. 

Any finite linear combination of the monomials 
$z^k\zzinfty^{l}$ is called a {\it Wick polynomial}. 
The word `polynomial' is to be understood 
as a polynomial 
expression of an element in $ \azz $ and not 
as an element in an abstract polynomial algebra. 
As we shall see two such distinct expressions can be equal. 

\begin{prop}\label{pointless} 
In the framework of Harmony Two, 
if $z$ and $\zzinfty$ satisfy that 
$\zzinfty z$ is expressible as 
a Wick polynomial, 
then the algebra $\azz$ has no characters.  
\end{prop}

\begin{proof}
Let us assume to the contrary, that 
$\varkappa\colon\azz\rightarrow\Bbb{C}$ 
is a character and define 
$w:=\varkappa(z) \in \Bbb{C}$. 
(Recall that here  $ z \in \azz $ is the operator 
of multiplication by $ z $ acting on $ \HolC $.)
Because $ \varkappa $ is a *-representation, 
we have 
$\varkappa(z^k\zzinfty^{l}) = w^k\bar{w}^l $. 
Since property H2 holds, 
by Proposition~\ref{nice-W-properties}
the point wave function $[w]$ 
belongs to the domain of all the polynomials in $Z$ 
and $ Z^* $. 
Moreover, for all $ k,l \in \Bbb{N} $ we have 
$$
\braket{[w]}{Z^kZ^{*l}[w]}=\braket{Z^{*k}[w]}{Z^{*l}[w]}=w^k\bar{w}^{l}\|[w]\|^2 = 
\varkappa(z^k\zzinfty^{l}) \|[w]\|^2.
$$ 
Write $\zzinfty z = p (z, \zzinfty) $, where $ p $ 
is a Wick polynomial. 
By using the previous equation and linearity
to get the second equality below, we then obtain 
$$
\braket{[w]}{Z^*Z[w]} = 
\braket{[w]}{p(Z,Z^*) [w]} = 
\varkappa(p(z,\zzinfty) ) \|[w]\|^2 =
\varkappa (\zzinfty z) \|[w]\|^2 = 
|w|^2  \|[w]\|^2.  
$$
The last equality holds because $ \varkappa $ 
is  a *-representation. 
Note that $ [w] $ is in the domain 
of the polynomial $ Z - w $. 
And hence we see that 
\begin{align*}
\|(z-w)[w]\|^2&=\braket{(Z-w)[w]}{(Z-w)[w]}
\\
&= \braket{[w]}{Z^*Z[w]} - \braket{Z[w]}{w[w]} 
- \braket{w[w]}{Z[w]} + \braket{w[w]}{w[w]}
\\
&= |w|^2  \|[w]\|^2 - |w|^2  \|[w]\|^2 
- |w|^2  \|[w]\|^2 + |w|^2  \|[w]\|^2 = 0.
\end{align*}
where in the two middle terms on the second line we used $ Z^* [w] = \bar{w} [w] $. 
In other words $(z-w)[w]=0$ in $\cal{H}$, 
which contradicts the fact 
that the holomorphic function
$ [w] $ is non-zero. 
\end{proof}

\begin{remark} The `flipping' condition given 
by a Wuck polynomial 
for the product $\zzinfty z$ is essential. 
The phenomenon is 
illustrated, in the last part of this section 
by an example where $\azz$ is an extension 
of the Manin $q$-plane by the infinite matrix 
algebra $\mathrm{M}_\infty(\Bbb{C})$.   
\end{remark}

\begin{definition} 
We shall say that 
{\it Harmony Three} (or simply H3) holds if 
both Topological H0 and H2 hold. 
\end{definition}
\noindent Harmony Three holds if and only if 
$1\in\W$, in the framework of Harmony Two.  

\begin{prop}\label{zzinfty-w(z)} 
In the framework of Harmony Three, 
the only eigenvectors of $\zzinfty$ in $\HolC$ are those of $Z^*$ in $\cal{H}$ or,  
in other words, 
the point wave functions $w(z)$ modulo scalar multiples. 
\end{prop}

\begin{proof} Let us assume that $\Psi(z)\in\HolC$ satisfies $\zzinfty\Psi(z)=\bar{w}\Psi(z)$ 
for some $w\in\Bbb{C}$. This means that $\braket{(z-w)\pp(z)}{\Psi(z)}=0$ for every $\pp(z)\in\Cz$. Here we 
have extended the scalar product by normal continuity to $\Cz\times\HolC$. The
extended $\braket{}{}$ is non-degenerate. This leaves only one degree of freedom for $\Psi(z)$. 
\end{proof}

\begin{prop} In the framework of Harmony Three, the commutant of 
$\azz$ in $\Opiz$ is trivial. In particular, the center of the algebra 
$\azz$ is trivial: $\mathrm{Z}(\azz)=\Bbb{C}$.   
\end{prop}

\begin{proof} If something from $\Opiz$ commutes with $z$ then 
it must be of a multiplicative form $\psi(z)\longmapsto \Theta(z)\psi(z)$, where $\Theta(z)\in\HolC$. 
If $\Theta(z)$ in addition commutes with $\zzinfty$ then $\zzinfty[\Theta(z)w(z)]=\bar{w}\Theta(z)w(z)$. 
Proposition~\ref{zzinfty-w(z)} implies that $\Theta(z)w(z)$ must be 
proportional to $w(z)$ which is possible only if $\Theta(z)$ is a constant function. 
\end{proof}

\subsection*{Projection Construction}

The following list recapitulates the symbols used for the diverse operators 
related intrinsically to the complex coordinate $z$. 

\smallskip
\begin{center}
{\it 
\begin{tabular}{l|r}
$z$ & The multiplication operator $\psi(z)\longmapsto z\psi(z)$ acting in $\Cz$, $\W$ or $\HolC$. \\
$Z$ & Induced $z$-multiplication in $\cal{H}$ with $\mathrm{D}(Z)=\bigl\{\psi(z)\in\cal{H}\bigm |
z\psi(z)\in\cal{H}\bigr\}$.  \\
$Z^*$ & The adjoint operator for the operator $Z$ in $\cal{H}$. \\ 
$\zzinfty$ & The continuous extension of $Z^*$ to the whole $\HolC$. \\
$z^*$ & The formal adjoint of the multiplication operator $z$ in $\Cz$ and $\W$. 
\end{tabular}
}
\end{center}

\smallskip
These symbols are all related to $\cal{H}$ and 
$\azz$. 
Following \cite{coherent}, we shall now interpret 
the algebra $\azz$ as an algebra of Toeplitz operators
for the Toeplitz quantization of another, a priori unrelated, algebra, which will be assumed to be a quadratic polynomial *-algebra generated by $z$ and $\bar{z}$. 
So $z$ is a true multi-personality creature here! 

The basic example is the standard commutative *-algebra $\Bbb{C}[z,\bar{z}]$ 
of polynomials in $z$ and $\bar{z}$. 
Its only defining relation is simply $z\bar{z}=\bar{z}z$, 
and its *-structure is the anti-linear map 
which exchanges $z$ and $\bar{z}$. 
The product between $\Bbb{C}[\bar{z}]$ and $\Cz $ 
induces a natural decomposition
\begin{equation}\label{basic-decomposition}
\Bbb{C}[z,\bar{z}]\leftrightarrow
\Bbb{C}[\bar{z}]\otimes\Cz  \qquad \bar{z}^n z^m\leftrightarrow \bar{z}^n\otimes z^m. 
\end{equation}
A variety of interesting non-commutative variations 
of this emerge, 
if we fix a hermitian $2\times 2$ matrix 
$$Q=\begin{pmatrix}q_{00} & q_{01}\\ q_{10} & q_{11}\end{pmatrix}
$$ 
so that $\overline{q_{10}}=q_{01}$ and $q_{00},q_{11}\in\Bbb{R}$,  
and define $\polq$ to be the *-algebra generated by $z$ and $\bar{z}$ together with the following relation: 
\begin{equation}\label{quantum}
z\bar{z}=\bar{z}z+(1\,\,\bar{z})\begin{pmatrix}q_{00} & q_{01}\\ q_{10} & q_{11}\end{pmatrix}
\begin{pmatrix}1\\z \end{pmatrix} 
\end{equation}
where also 
$q_{11}\neq -1$. The *-structure, as in the commutative case, 
is anti-linear and 
exchanges $z$ and $\bar{z}$. 
We can think of $Q$ as being a quadratic form on $\Bbb{C}^2$.
We recover the commutative case if $Q=0$.  

These non-commutative polynomials keep valid 
the basic decomposition \eqref{basic-decomposition}
or, in other words for every $Q$ we have 
an isomorphism of vector spaces, 
but not of algebras
\begin{equation}\label{basic-decomposition-Q}
\polq\leftrightarrow
\Bbb{C}[\bar{z}]\otimes\Cz  \qquad 
\bar{z}^n z^m\leftrightarrow \bar{z}^n\otimes z^m 
\end{equation}
by the composition-diamond lemma. 
(See \cite{shepler}.)  
In terms of this identification 
$1\otimes 1$ is the unit element of the algebra, 
and $\Cz $ 
and $\Bbb{C}[\bar{z}]$ are mutually conjugate unital subalgebras of $\polq$.  
The basis elements $ \bar{z}^n z^m $ of $ \polq $ 
are called {\it anti-Wick monomials}. 
Consequently, every element in $ \polq $ has a 
unique expression as a {\it anti-Wick polynomial}, 
that is as a finite linear combination of the 
anti-Wick monomials. 

\begin{remark}
The coordinates $\bar{z}$ and $z$ 
appear symmetrically at this point of the theory, 
since the condition $q_{11}\neq -1$ ensures 
the presence of non-zero coefficients of 
both $z\bar{z}$ and $\bar{z}z$ 
in the defining relation \eqref{quantum}. 
Then the  composition-diamond lemma  
implies that there is also an `opposite' 
vector space isomorphism, but not 
an algebra isomorphism,   
$\polq \leftrightarrow \Cz\otimes\Bbb{C}[\bar{z}]$ 
given by 
$ z^n \bar{z}^m\leftrightarrow z^n\otimes \bar{z}^m  $.
The basis elements $ z^n \bar{z}^m $ of $ \polq $ 
are called {\it Wick monomials}.  
Symmetrically, every element in $ \polq $ is uniquely 
expressible as a Wick polynomial. 
\end{remark}

It follows that the product of $ \polq $
is encoded in the `flipping' rules for monomials $z^n$ and $\bar{z}^m$ which 
can be written as finite sums 
\begin{equation}\label{4-array}
z^n \bar{z}^m=\sum_{k,l\geq 0}\qqqq{n}{m}{k}{l}\,\bar{z}^l z^k
\end{equation}
for an array with 4 indices 
$\qqqq{n}{m}{k}{l}$ 
of complex numbers which 
can only have non-zero values 
for $0\leq k\leq n$ and $0\leq l\leq m$, 
since the relation \eqref{quantum} has 
no terms of order
greater than $ 2 $ in $ z $ and $ \bar{z} $. 
Polynomial expressions whose terms are all monomials of 
the form $ \bar{z}^l z^k $ as in \eqref{4-array} 
are said to be {\em reduced} in abstract ring theory and in {\em anti-Wick order} in quantum theory. 

This array must obey a number of consistency conditions
in order to represent a unital 
associative *-algebra 
structure on $\polq$. At first, by conjugating the above formula we arrive at  
\begin{equation}\label{4-hermicity}
\overline{\qqqq{n}{m}{k}{l}}=\qqqq{m}{n}{l}{k}.  
\end{equation}
Then, the requirement for the existence of a 
unit translates into
\begin{equation}
\qqqq{n}{0}{k}{l}=\delta_{nk}\delta_{0l}\qquad \qqqq{0}{m}{k}{l}=\delta_{0k}\delta_{ml}. 
\end{equation}
More generally, $z^n$ and $\bar{z}^m$ commute if and only if 
\begin{equation}
\label{4-array-commutes-condition}
\qqqq{n}{m}{k}{l}=\delta_{nk}\delta_{ml}
\end{equation}
for all $k,l\geq 0$.
Finally, the following charming and easy to remember 
convolution identities hold, where $ \times $
denotes the product of complex numbers:  
\begin{equation}\label{convolution}
\qqqq{u+v}{m}{k}{l}=\sum_{\substack{s\geq 0 \\a+b=k}} \qqqq{u}{s}{a}{l}\times \qqqq{v}{m}{b}{s}\qquad\qquad
\qqqq{n}{u+v}{k}{l}=\sum_{\substack{s\geq 0 \\a+b=l}} \qqqq{n}{u}{s}{a}\times \qqqq{s}{v}{k}{b}.
\end{equation}
They reflect the fact that both $\Cz $ and $\Bbb{C}[\bar{z}]$ are unital subalgebras of $\Bbb{C}[z,\bar{z}]$. 
In particular, by inductively applying these identities, we see that the array 
$\qqqq{n}{m}{k}{l}$ is completely determined by $Q$ 
because
\begin{equation}
\qqqq{1}{1}{0}{0}=q_{00}\qquad \qqqq{1}{1}{0}{1}=q_{01}\qquad \qqqq{1}{1}{1}{0}=q_{10}\qquad
\qqqq{1}{1}{1}{1}=1+q_{11}
\end{equation}
which is the same as the basic generating relation \eqref{quantum}.  

In order to clarify different roles of the same space, 
we shall write $\pinfty$ for $\polq$ as a vector space on which the same  
algebra acts via its left 
regular representation. In a similar manner, we shall write $\cal{P}$ for $\Cz$ understood as the corresponding representation space. 
Clearly $\cal{P}$ is a subalgebra of $\pinfty$. 
Also, every element in  $\pinfty$ can be 
written uniquely as a finite sum of
polynomials homogeneous in $\bar{z}$, namely as 
$ \sum_{n} a_{n} \bar{z}^n \varphi_{n} $ 
with $ a_{n} \in \Bbb{C} $ and 
$ \varphi_{n} \in \cal{P}$. 
Singling out the subalgebra $\cal{P} = \Cz$ 
for special consideration, instead of 
$\Bbb{C}[\bar{z}]$ which we could have done, 
breaks the previously mentioned symmetry 
between $ z $ and $ \bar{z} $. 

We shall assume here that H3 holds for $\cal{H}$. 

\begin{prop}
There exists a unique quadratic form $\qform{}{}$ on $\pinfty$ which extends the scalar product on 
$\cal{P}$ and such that the left regular representation 
of $ \polq $ on $\pinfty$ 
is symmetric with respect to that form. 
It is determined by
\begin{equation}\label{q-forma}
\qform{\bar{z}^n \varphi}{\bar{z}^m \psi}=\sum_{k,l\geq 0}\qqqq{n}{m}{k}{l}\braket{Z^{l}\varphi}{Z^{k}\psi}
\end{equation}
for $n,m\in\Bbb{N}$ and $\varphi,\psi\in\cal{P}$.  
\end{prop}

\begin{proof} Let us suppose that $\qform{}{}$ exists. Then
\begin{equation*}
\qform{\bar{z}^n\varphi}{\bar{z}^m\psi}=
\qform{\varphi}{z^n\bar{z}^m\psi}=\sum_{k,l\geq 0}\qqqq{n}{m}{k}{l}\qform{\varphi}{\bar{z}^lz^k\psi}
=\sum_{k,l\geq 0}\qqqq{n}{m}{k}{l}\braket{Z^l\varphi}{Z^{k}\psi}
\end{equation*}
so indeed \eqref{q-forma} holds and hence $\qform{}{}$ must be unique. 
To prove its existence, we define 
it on $\bar{z}$-homogenous polynomials by \eqref{q-forma}, extending it to $\pinfty\times\pinfty$ by using bi-additivity. From $\qqqq{0}{0}{k}{l}=\delta_{k0}\delta_{l0}$ it follows that 
$\qform{}{}$ extends $\braket{}{}$. 
It is also hermitian symmetric because of 
\eqref{4-hermicity}:
$$ 
\qform{\bar{z}^n \varphi}{\bar{z}^m \psi}^*=
\sum_{k,l\geq 0}\overline{\qqqq{n}{m}{k}{l}}\braket{Z^{k}\psi}{Z^{l}\varphi}=
\sum_{k,l\geq 0}\qqqq{m}{n}{l}{k}\braket{Z^{k}\psi}{Z^{l}\varphi}=\qform{\bar{z}^m \psi}{\bar{z}^n\varphi}. 
$$ 
So we have a quadratic form on $\pinfty$. We next have to prove that $z$ and $\bar{z}$ are mutually formally
adjoint with respect to $\qform{}{}$. 
Again, we verify this on polynomials homogeneous in $\bar{z}$
by using \eqref{4-array}, \eqref{convolution} 
and \eqref{q-forma}: 
\begin{multline*}
\qform{\bar{z}^n\varphi}{\bar{z}\bar{z}^m\psi}=\sum_{k,l\geq 0}\qqqq{n}{m+1}{k}{l}\braket{Z^l\varphi}{Z^k\psi}=
\!\!\!\sum_{k,s,a,b\geq 0}\!\{\qqqq{n}{1}{s}{a}\times\qqqq{s}{m}{k}{b}\}\braket{Z^{a+b}\varphi}{Z^k\psi}=\\
=\!\!\!\sum_{k,s,a,b\geq 0}\!\!\!\qqqq{s}{m}{k}{b}\,\braket{Z^b\bigl\{\qqqq{1}{n}{a}{s}\,Z^a\varphi  \bigr\}}{Z^k\psi}=
\!\!\sum_{a,s\geq 0}\qform{\bar{z}^s\bigl\{ 
\qqqq{1}{n}{a}{s} Z^a\varphi\bigr\}}{\bar{z}^m\psi}=\qform{z\bar{z}^n\varphi}{\bar{z}^m\psi}. 
\end{multline*}
So $z$ and $\bar{z}$ are formally adjoint relative to $\qform{}{}$. Since they generate $\polq$ this property extends to the whole polynomial algebra.  
In other words we get that 
\begin{equation}
\qform{F\varphi}{\psi}=\qform{\varphi}{\overline{F}\psi}
\end{equation} 
for every $F\in\polq$ and $\varphi,\psi\in\pinfty$. 
\end{proof}

\begin{remark} The form $\qform{}{}$ will not be positive in general. 
The positivity problem for 
the form, in the context of the commutative algebra $\Bbb{C}[z,\bar{z}]$, relates to the {\it moment problem} --- the existence of a finite measure $\mu$ on $\Bbb{C}$ reproducing the scalar product via the standard $L^2$ inner product 
\begin{equation}
\qform{f}{g}=\int_\Bbb{C}\bar{f}g\, d\mu(z,\bar{z}).
\end{equation}
It is easy to see that if a measure $\mu$ on $\Bbb{C}$ reproduces the scalar product $\braket{}{}$ on $\cal{P}$, 
then it automatically reproduces the extended form $\qform{}{}$ on $\Bbb{C}[z,\bar{z}]$, and in particular it 
will be positive (possibly not strictly positive). On the other hand there exist interesting examples where 
$\qform{}{}$ is strictly positive, 
even though there will be no underlying measure. 
As we shall see, such 
exotic situations do not occur if the scalar product 
$\braket{}{}$ on $\cal{P}$ is diagonal. 
\end{remark}

In resonance with this, we can introduce an 
integration functional for the algebra $\polq$. 
We shall use the same 
generic integration symbol as for $\azz$, and similarly to \eqref{integral-A} define 
\begin{equation}
\int F=\qform{0}{F}. \quad \forall  F \in \polq. 
\end{equation} 
(Recall the Dirac notation is $
 \ket{0} = \pp_{0}(z) = 1 \in \polq$.)
This is a hermitian functional evaluating to 
$s_{nm}=\braket{z^n}{z^m}$ on the monomials 
$\bar{z}^nz^m$. 
Moreover, 
\begin{equation}\label{qform-Phi}
\qform{\rho}{\pp}=\int \overline{\rho}\,\pp
\end{equation}
for every $\rho,\pp\in\polq$. 

\begin{prop} 
\label{PI-proposition}	
The following conditions are equivalent: 

\noindent---The space $\cal{P}$ is orthocomplementable in $\pinfty$ or in other words there 
exists a subspace $ \cal{P}^\bot $ such that 
\begin{equation}\label{p-pbot}
\pinfty=\cal{P}\oplus\cal{P}^\bot.
\end{equation}

\noindent---There exists a linear map $\Pi\colon\pinfty\rightarrow\pinfty$ satisfying 
\begin{equation}
\label{PI-properties}
\Pi^2=\Pi\qquad \im(\Pi)=\cal{P}\qquad\qform{\Pi\psi}{\varphi}=\qform{\psi}{\Pi\varphi}
\end{equation}
for every $\varphi,\psi\in\pinfty$. 

\noindent--- Algebraic H0 holds. In other words $\cal{P}$ is $Z^*$-invariant and $Z^*\vert\cal{P}=\zz$ the formal adjoint of $z$ in $\cal{P}$. 

\smallskip
If any of these (and hence all of these) hold, 
then $\ker(\Pi)=\cal{P}^\bot$ and in particular $\Pi$ is the projection
associated to the orthogonal decomposition of $\pinfty$. 
Moreover, there is a commutative diagram 
\begin{equation}\label{z*-def2}
\begin{CD}
\pinfty @>{\mbox{$\bar{z}$}}>> \pinfty\\
@AAA  @VV{\mbox{$\Pi$}}V\\
\cal{P} @>>{\mbox{$z^*$}}> \cal{P}, 
\end{CD}
\end{equation}
where the left vertical arrow 
is the inclusion map and the top horizontal arrow 
is multiplication by $ \bar{z} $. 
\end{prop}

\begin{proof}
If \eqref{p-pbot} holds, then $\Pi$ is 
defined to be the associated projection on $\cal{P}$ and \eqref{PI-properties} 
is immediate. 
On the other hand
if \eqref{PI-properties}  holds, then 
the kernel of $ \Pi $ is an orthocomplement 
of $\cal{P}$ and so \eqref{p-pbot} holds. 

If $\Pi$ exists as in \eqref{PI-properties},
we define $z^*$ by \eqref{z*-def2}.
Then for $\varphi,\psi\in\cal{P}$ we have 
$$ \braket{\varphi}{z^*\psi}=\braket{\varphi}{\Pi\bar{z}\psi}=
\qform{\Pi\varphi}{\bar{z}\psi}=\qform{\varphi}{\bar{z}\psi}=\qform{z\varphi}{\psi}=\braket{z\varphi}{\psi}. 
$$
So the operator $z^*$ 
is indeed the formal adjoint of $z$ in $\cal{P}$ and so 
Algebraic H0 holds. 

Finally, assume Algebraic H0. 
So the operator $z^*$ exists 
(i.e., $Z^*$ preserves $\cal{P}$). 
Then we can define 
$\Pi$ on $\bar{z}$-homogeneous polynomials by 
\begin{equation}\label{Pi-def}
\Pi(\bar{z}^n\psi)=z^{*n}\psi \in \cal{P}
\end{equation}
for $\psi\in\cal{P}$ and $n\in\Bbb{N}$, and extend it 
to $\pinfty$ by additivity. 
By construction $ \Pi $ is an idempotent projecting on $\cal{P}$. It is also symmetric. 
Indeed if 
$\varphi$ is also from $\cal{P}$, then 
\begin{multline*}
\qform{\Pi(\bar{z}^n\varphi)}{\bar{z}^m\psi}=\qform{z^{*n}\varphi}{ \bar{z}^m\psi}=\braket{z^mz^{*n}\varphi}{ \psi}
=\braket{z^{*n}\varphi}{ z^{*m}\psi}=\\
{}=\braket{\varphi}{ z^nz^{*m}\psi}=\qform{\bar{z}^n
\varphi}{ z^{*m}\psi}=\qform{\bar{z}^n\varphi}{\Pi(\bar{z}^m\psi)}
\end{multline*}
and by bi-additivity the identity extends to the whole $\pinfty$. 
Therefore all the identities in \eqref{PI-properties} 
have been proved. 
\end{proof}

Here is an explicit way to calculate the projection map $\Pi$ in terms of the 
canonical orthonormal polynomials 
$\pp_n(z) \in \Cz \subset \polq$. 

\begin{prop}
Assume that $ \Pi $ exists as in 
Proposition~\ref{PI-proposition}. 
We then have for every $F\in\polq$ that
\begin{equation}\label{Pi-B}
\Pi(F)=
\sum_{k=0}^{\infty} \qform{\pp_k(z)}{F}\pp_k(z),  
\end{equation}
There are only finitely many non-zero terms 
in this infinite sum. 
\end{prop}
\begin{proof}
For $ F \in \cal{P} = \Cz $ equation 
\eqref{Pi-B} reduces to saying that 
the polynomials $ \pp_n(z) $ are an orthonormal basis 
of the vector space $ \Cz $, in which case only finitely many terms are non-zero.  
For $F \in \cal{P}^\bot = \ker(\Pi)$ the left 
side of \eqref{Pi-B} is zero and every term 
on the right side is also zero since 
$ \pp_n(z) \in \cal{P}  $. 
Then \eqref{Pi-B} follows for every  
$F\in\cal{P}\oplus\cal{P}^\bot=\polq$ 
by additivity. 
\end{proof}
\begin{remark}
In Dirac notation (which technically does not apply,
since we are not in a Hilbert space setting) 
we can write \eqref{Pi-B} as 
\begin{equation}\label{Pi-diadas}
\Pi = \sum_{k=0}^{\infty} |k\rqbra\lqbra k|.  
\end{equation}
\end{remark}

The quadratic form $\qform{}{}$ will fail 
in general to be non-degenerate. 
In other words, we might 
encounter a non-trivial 
null-space $\ninfty=\pinfty^\bot$. 
This next result is immediate. 

\begin{prop} The null-space $\ninfty$ is an invariant subspace for the left regular representation of $\polq$ in $\pinfty$.
Moreover $\Pi$, if it exists, maps $\ninfty$ into $\{0\}$.
Thus, the whole module structure, 
quadratic form and $\Pi$ naturally project down to the factor space $\qinfty=\pinfty/\ninfty$,
preserving all the 
basic formulas. 
The projected $\qform{}{}$ is non-degenerate. We have 
$\cal{P}\cap\ninfty=\{0\}$, 
and so $\cal{P}$ is naturally a subspace of $\qinfty$. \qed
\end{prop}

\begin{definition}
For any $f \in \polq$
we define the {\em Toeplitz operator with symbol $f$}, denoted as
$ \Tz(f)=\Tz_f : \mathcal{P} \to \mathcal{P} $, by $ \Tz_f\pp := \Pi (f \pp) $ for all  
$ \pp \in \mathcal{P} $.  
Notice that the product $f \pp $ of the 
two elements $f, \pp \in \polq$ 
is again an element in the algebra 
$\polq$. Then the projection $\Pi$ maps this product to 
an element of $ \mathcal{P}$. 
In this way we obtain a linear map
$\Tz\colon\polq\rightarrow \azz$, 
the *-algebra of operators generated by $z$ and $\zz$. 
We say that $ \Tz $ is the 
{\em Toeplitz quantization}.  
\end{definition}

Even though $ \Tz $ is a linear map from one algebra
to another, 
it is not expected nor desired to be 
multiplicative, that is, a map of algebras. 
We will come back to this point. 
However,  $ \Tz $ does preserve the identity element,
namely  $ \Tz_{1} $ is the identity operator 
of $ \mathcal{P} $. 

As a direct consequence of the above definition of $\Tz$, we find that 
\begin{equation}\label{T-z-Z}
\Tz(\bar{z}^nz^m)= \zz{}^n z^m. 
\end{equation}
So the recipe to calculate the operators $\Tz_f$ is quite simple 
and can be used as an alternative definition of $\Tz$. 
Just replace $z$ and $\bar{z}$ by their counterparts $z$ and $\zz$ in $\azz$
in the polynomial expression for $f$, assuming that $f$ is written in anti-Wick form, that is,   
all $\bar{z}$'s are moved to the left of 
$z$'s. 
This is the reason for saying 
that Toeplitz quantization 
is an {\it anti-Wick quantization}. 
Equation \eqref{T-z-Z} works out so nicely in 
part because the definition of $ \zz $ 
in diagram \eqref{z*-def2} now reads as 
$ \Tz(\bar{z})= \zz{} $. 
On the other hand the identity 
$ \Tz(z)= z $ follows immediately from 
the fact that $ \cal{P} = \Cz$ is 
$ z $-invariant. 
Equation \eqref{T-z-Z} also shows that all the operators $\Tz_f$ are indeed in $\azz$. 

\begin{remark} In general, $\Tz$ will not be surjective. Its image is the $\Bbb{C}[\bar{z}]$--$\Cz$ bimodule
in $\azz$ generated by $1\in\azz$. So $\azz$ is always generated by the image of $\Tz$. In this context, we shall 
refer to $\azz$ as the {\it algebra of Toeplitz operators}.  
\end{remark}

If we trivially 
extend the operators in $\cal{P}$ to $\qinfty$ by requiring that they vanish on $\cal{P}^\bot$, 
then we can write 
\begin{equation}\label{compression-by-Pi}
\Tz_f=\Pi f\Pi
\end{equation}
where on the right side $f$ is interpreted as the left regular representation operator. 
The expression on the right side of 
\eqref{compression-by-Pi} is known as the 
{\it compression} by the projection $ \Pi $
of the operator defined by $ f $

We see that $\Tz$ intertwines the 
*-structures on $\polq$ and $\azz$. It also connects the integration functionals on both algebras. 
\begin{prop} We have 
\begin{equation}\label{integration-2}
\int \Tz=\int. 
\end{equation}
\end{prop}

\begin{proof} It is sufficient to check the identity 
on the monomials $\bar{z}^n z^m$. 
And then both the left and the right hand 
side evaluate to $\braket{z^n}{z^m}=s_{nm}$. 
\end{proof}

\begin{lemma} For every $\rho,\pp\in\cal{P}$ and $n\in\Bbb{N}$ we have 
\begin{equation}
\qform{\bar{z}^n\rho}{(\bar{z}-\zz)\pp}=\braket{\rho}{\Tz(z^n\bar{z})\pp}-\braket{\rho}{z^n\zz\pp}. 
\end{equation}
In particular, the map $\Tz$ is multiplicative only in the trivial scenario $\qinfty=\cal{P}$. 
\end{lemma}

\begin{proof}
We compute 
\begin{align*}
\qform{\bar{z}^n\rho}{(\bar{z}-\zz)\pp}
&=\qform{\bar{z}^n\rho}{\bar{z}\pp}-\qform{\bar{z}^n\rho}{\zz\pp}
=\qform{\rho}{z^n\bar{z}\pp}-\braket{\rho}{z^n\zz\pp} 
\\
&=\braket{\rho}{\Tz(z^n\bar{z})\pp}-\braket{\rho}{z^n\zz\pp}. 
\end{align*}
We see that $\Tz(z^n\bar{z})=z^n\zz$ if and pnly if 
$\bar{z}$ acts as $\zz$ on $\cal{P}$. Since $\cal{P}$ is 
cyclic for $\bar{z}$ in $\qinfty$ this is only possible when the two spaces coincide.  
\end{proof}

\begin{remark}
Therefore, the complement $\cal{P}^\bot$ of $\cal{P}$ in $\qinfty$ can be viewed as a subtle measure of difference between $\polq$ and $\azz$. It is also worth observing that $(\bar{z}-\zz)\pp\in \cal{P}^\bot$ always.  
\end{remark}

\subsection*{The Fifth Element} 

To sum things up so far, 
our construction produces a canonical $\polq$ representation space $\qinfty$ extending $\cal{P}$ and equipped with a non-degenerate hermitian form $\qform{}{}$. If Algebraic H0 holds, there is a symmetric idempotent $\Pi$ projecting $\qinfty$ onto $\cal{P}$ and realizing the operators $z$ and $\zz$ of $\azz$ as 
compressions of $z$ and $\bar{z}$ of $\polq$ 
by $\Pi$. This is our main algebro-geometric 
setting for the Toeplitz 
quantization. 

The whole construction can be performed without essential changes with $\W$ instead of $\cal{P}$. The advantage 
is that $\W$ is always invariant under both $z$ and $\zzinfty$, as far as we stay within Harmony Two. In this case 
the space $\pinfty$ is redefined as 
\begin{equation}
\pinfty=\polq\otimes_{\Cz}\W\leftrightarrow\mathbb{C}[\bar{z}]\otimes\W. 
\end{equation}

This brings us naturally to our final harmony property.
It is about a mutual relationship 
between the two principal algebras $\polq$ and $\azz$, related via the Toeplitz quantization. 

\begin{definition}
We shall say that {\it Harmony Four} 
 (or simply H4) 
holds if the form $\qform{}{}$ is positive. 
\end{definition}

In the framework of H4 the space $\qinfty$ 
is a pre-Hilbert space and 
can be completed into a Hilbert space $\cal{J}$.  Moreover, $\cal{H}$ 
can be viewed as a subspace of $\cal{J}$. 
The projection $\Pi$ extends 
to an orthogonal projection $\Pi\colon
\cal{J}\rightarrow\cal{H}$, where we use 
the same notation. 

\begin{remark}
This is all much in the spirit of the Stinespring 
construction \cite{Sti} for completely positive maps between C*-algebras. 
\end{remark}

\subsection*{Classification}

The algebra $\polq$ falls into one of 
four distinguished {\it classes}, corresponding to  four canonical forms of $Q$, obtained after making a linear substitution $z\rightsquigarrow az+b$ 
where $a,b\in\Bbb{C}$ and $a\neq 0$. 
Under such a substitution the matrix $Q$ transforms as 
\begin{equation}
\begin{pmatrix}
q_{00} & q_{01}\\
q_{10} & q_{11}
\end{pmatrix}
\rightsquigarrow 
\begin{pmatrix}
q(\bar{b},b)/\vert a\vert^2 & (q_{01}+q_{11}\bar{b})/\bar{a}\\
(q_{10}+q_{11}b)/a & q_{11}
\end{pmatrix}
\end{equation}
where $q(\bar{b},b)=q_{00}+q_{01}b+q_{10}\bar{b}+q_{11}|b|^2$ 
as follows from the defining relation \eqref{quantum}. 
Let us briefly describe these classes. 

\subsubsection*{Principal Flipping Type} This is defined by $q=1+q_{11}\neq 0, 1$. 
After the appropriate linear substitution, we arrive at 
\begin{equation}
z\bar{z}=q\bar{z}z + h\qquad h\in\{-|q|,0,1\}\qquad Q=\begin{pmatrix}h & 0\\ 
0 & q_{11}\end{pmatrix}. 
\end{equation} 

There are 6 important special subcases. The first one is given by $h=0$. This is the Manin $q$-plane. A detailed analysis of this example and the general theory of its Toeplitz quantization can be found in \cite{coherent}.  
 
The classical part of the corresponding quantum space (given by the characters of $\polq$) is just one point--the character evaluating to 0 on $z$ and $\bar{z}$. 
An important part of the analysis is the question of realizability of the algebra by operators in a Hilbert space, the case $q>0$ is realizable, and we can always assume $0<q<1$ since 
if $q>1$ then by interchanging $z$ and $\bar{z}$ we have $q\rightsquigarrow 1/q$. As we shall explain below, 
there is a natural realization of this algebra in a Hilbert space of holomorphic functions in $\Bbb{C}-\{0\}$. The case 
$q<0$ is, clearly, not realizable by operators in a Hilbert space.  

Next, there are two essentially different cases with $q>0$ and $h=-q,1$. If $0<q<1$ and $h=-q$ or $q>1$ and $h=1$ then we obtain equivalent realizations of a $q$-variant of the standard quantum plane. If, on the other hand, $0<q<1$ and $h=1$ or equivalently $q>1$ and $h=-q$, then the algebra will represent a Poincar\'e model for a quantum hyperbolic plane \cite{MP}. The horizon of infinity is traced by a classical circle centered at $0$ of radius $\sqrt{h/1-q}$. These can be morphed into 
$$\bar{z}z=qz\bar{z}+1-q$$ 
with $0<q<1$, to fit the unitary disk $\Bbb{D}$. 

If $q$ is negative, there are also two essentially different situations. The first one is given by $h=q$ and corresponds 
to a character-free *-algebra non-realizable in a Hilbert space. If $h=1$ then we are again within the quantum 
hyperbolic planes, the structure can be morphed into the same generating expression as above, but now with $-1<q<0$. 

\subsubsection*{Parabolic Type} This corresponds to $q_{11}=0$ with $q_{01}=\overline{q_{10}}\neq 0$, which can be transformed into
\begin{equation}\label{parabolic-z-z*}
\bar{z}z=z\bar{z}+z+\bar{z}\qquad\quad Q=-\begin{pmatrix}0 & 1\\1 & 0\end{pmatrix}. 
\end{equation}
The classical points are naturally labeled by purely imaginary numbers. The generators $z$ and $\bar{z}$ 
linearly span a non-commutative 2-dimensional Lie algebra. Abstractly, there is only one such a structure 
and one concrete realization is the (complexified) Lie algebra of the group $t\longmapsto at+b$ where 
$a,b\in\Bbb{R}$ with $a>0$. These are 
orientation preserving affine transformations of $\Bbb{R}$. The algebra $\polq$ 
is then viewable as the universal envelope of this Lie algebra.  

\subsubsection*{Orthodox Quantum Plane} This corresponds to $q_{11}=q_{01}=q_{10}=0$ and $q_{00}\neq 0$. In 
this case $q_{00}$ can be scaled to $-1$ or $1$ and we obtain the Heisenberg-Weyl algebra 
\begin{equation}
z\bar{z}-\bar{z}z=\pm 1\qquad\quad Q=\pm\begin{pmatrix}1 & 0 \\ 0 & 0\end{pmatrix}. 
\end{equation} 

\subsubsection*{Classical Euclidean Plane} 
The classical commutative polynomial algebra $\Bbb{C}[z,\bar{z}]$ is obtained 
when $q_{11}=q_{01}=q_{10}=q_{00}=0$. 

\subsection*{Diagonal Scalar Product}

In this sub-section we will assume that 
the initial scalar product is diagonal, 
in which case 
all calculations significantly simplify. 
As we have explained already, 
property H0 always holds in this context. 
Also, recall the definition \eqref{simple-sequence}
of the sequence 
$ s_{n} := s_{nn} = \braket{z^n}{z^n} > 0 $, which 
characterizes $ \braket{}{} $ in this case. 
\begin{lemma}
In terms of the sequence $ s_{n} $, 
the quadratic form $\qform{}{}$ is 
determined by  
\begin{equation}\label{form-general}
\qform{\bar{z}^i z^j}{ \bar{z}^n z^m}=\sum_{k\geq 0}\qqqq{i}{n}{k}{k+m-j}s_{k+m}. 
\end{equation}
\end{lemma}
\begin{proof}
This is straightforward from 
the definitions of $ s_{n} $ and 
the 4-array $\qqqq{}{}{}{}$: 
\begin{align*}
\qform{\bar{z}^i z^j}{\bar{z}^n z^m}
&=\qform{z^j}{z^i\bar{z}^nz^m}
=\sum_{k,l\geq 0} \qqqq{i}{n}{k}{l}\langle z^{j+l},z^{k+m}\rangle=\sum_{k,l\geq 0} \qqqq{i}{n}{k}{l}
s_{j+l}\delta_{j+l,k+m}
\\
&=\sum_{k\geq 0}\qqqq{i}{n}{k}{k+m-j}s_{k+m}. 
\end{align*}
We have also used the fact that the *-operation exchanges $z$ and $\bar{z}$. 
\end{proof}

\begin{remark}
In particular the quadratic form for the commutative polynomial algebra $\Bbb{C}[z,\bar{z}]$ 
is determined by
\begin{equation}
\label{form-extended}
\qform{
\bar{z}^iz^{j}}{ 
\bar{z}^{n}z^m}= s_{i+m}\, \delta_{i-j,n-m},
\end{equation}
where we used \eqref{4-array-commutes-condition}. 
This also follows easily by direct verification 
without appealing to the above Lemma. 
Notice even in this highly simplified case
that the vector space basis $ \bar{z}^{n}z^m $ 
of $\Bbb{C}[z,\bar{z}]$ is not orthogonal. 
Formula \eqref{form-extended} with 
$ s_{n} = n! $ is readily available in
the setting of Bargmann's seminal paper \cite{barg}, 
even though it is not explicitly given there.     
\end{remark}

Let us now discuss Harmony One and Two. 
As for the interpretability of all the elements 
of $\cal{H}$ as entire functions, 
this is not always possible.
It requires a special asymptotic 
behavior of the 
$\braket{}{}$ defining sequence $ s_{n} $.  

\begin{prop} A necessary and sufficient condition for Harmony One is that 
\begin{equation}\label{s-infinity}
\lim \, s_n^{1/n}=+\infty. 
\end{equation}
In this case the reproducing kernel for $\cal{H}$ is given by
\begin{equation}\label{K-diagonal-case}
K(\bar{w},z)=\sum_{n\geq 0} \frac{\bar{w}^n z^n}{s_n}. 
\end{equation}
\end{prop}

\begin{proof}
The formula \eqref{K-diagonal-case} 
for the reproducing kernel is a special case of \eqref{K} with $\pp_n(z)$ given by \eqref{pp-n-s}. 
Let us now apply Proposition~\ref{harmony-one}. 
The series \eqref{z-ppn-series} becomes 
$$ \sum_{n=0}^\infty |\pp_n(z)|^2=\sum_{n=0}^\infty \frac{|z|^{2 n}}{s_n}. $$

For this to be normally convergent on $\Bbb{C}$ a necessary and sufficient condition is indeed  \eqref{s-infinity}, 
as the Cauchy-Hadamard formula reveals. 
Next, the $\infty$-linear independence always holds
here, because all power series encode the values of their coefficients. 
\end{proof}

\begin{prop} A necessary and sufficient condition for Harmony Two, in the framework of Harmony One, is that 
the sequence   
\begin{equation}\label{tn-infty}
t_n=\Bigl(\frac{s_{n+1}}{s_n}\Bigr)^{1/n}
\end{equation} 
be bounded. 
\end{prop}

\begin{proof} If $Z^*$ is normally continuous by H2, 
then by using \eqref{pp-n-s} and \eqref{Z-down} 
we have  
$$
\sum_{n\geq 0}c_nz^n\longmapsto \sum_{n\geq 0}(c_{n+1}s_{n+1}/s_n)z^n
$$
consistently defines the extension of $Z^*$ on the whole $\HolC$. 
This means that the radius of convergence for  the resulting power series on the right must 
always be infinite, 
that is, the same as the radius of convergence of the initial power series on the left. 
This property can be rephrased in terms of the sequences as follows. For any sequence of complex
numbers $r_n$ converging to $0$, the sequence $r_nt_n$ must also converge to zero. 
This means that the sequence $t_n$ belongs to the multiplier algebra of the algebra $\mathrm{C}_0(\Bbb{N})$
of all sequences having limit zero. 
This multiplier algebra is precisely 
all bounded sequences $\mathrm{B}(\Bbb{N})$. So $t_n$ must be bounded. 

Conversely, it is a matter of a direct verification that the boundedness of $t_n$ implies that 
the above formula defines a (necessarily unique) continuous extension of $Z^*$ on the whole $\HolC$.  
\end{proof}

\begin{prop}
A particularly effective scenario for the applicability of the above criterion occurs when 
\begin{equation}\label{2-razlomka}
\lim s_{n+1}/s_n=+\infty\qquad\quad\lim s_n s_{n+2}/s_{n+1}^2 <+\infty. 
\end{equation}
\end{prop}

\begin{proof} A direct application of the classical convergence criterion. 
If $s_{n+1}/s_n$ is convergent then $s_n^{1/n}$ is convergent with the same limit. Similarly if 
the second sequence is convergent then $t_n$ will be convergent with the same limit.  
\end{proof}

\subsection*{The Two Frameworks Intersection}

Let us now elaborate more on the special situation, interesting in its own light, 
where the quadratic algebra $\polq$ is itself viewable as the Toeplitz operator algebra 
$\azz$ for some Hilbert space of entire functions $\cal{H}$ satisfying Harmony Two. 

The Heisenberg-Weyl algebra can be viewed in this way, via 
the Segal-Bargmann space, where $s_n=n!$ and $K(\bar{w},z)=\exp(\bar{w}z)$. 
(See \cite{barg}.)
Another class of examples is given by the $q$-variation 
\begin{equation}
z\bar{z}=q\bar{z}z-q\qquad\quad 0<q<1
\end{equation}
of the Heisenberg-Weyl algebra. 
Indeed, if we define the scalar product by the sequence 
\begin{equation}\label{q-factorial}
s_n=!_q(n):=\prod_{k=0}^{n-1}(1+1/q+\cdots+1/q^{k-1})
\end{equation}
then 
$$
s_{n+1}/s_{n}=1+1/q+\cdots+1/q^n\rightsquigarrow\infty\qquad\quad \frac{s_ns_{n+2}}{s_{n+1}^2}
=\frac{1}{q}\frac{1-q^{n+2}}{1-q^{n+1}}\rightsquigarrow \frac{1}{q}
$$ 
and thus H2 holds. A direct calculation then reveals that $z\longmapsto Z$ and $\bar{z}\longmapsto
Z^*$ extends to an isomorphism between $\polq$ and $\azz$. 

Furthermore, it turns out that the corresponding reproducing kernel is given by 
\begin{equation}\label{K=Eq}
K(\bar{w},z)=E_q(\bar{w}z)
\end{equation}
where $E_q$ is the $q$-exponential function. Let us recall (\cite{BHS}--Appendix II) 
that the $q$-exponential function is defined as  
\begin{equation}
E_q(z)=(qz-z|q)_\infty= \sum_{n=0}^\infty \frac{(1-q)^n}{(q|q)_n}q^{\textstyle{\binom{n}{2}}} z^n
=\sum_{n=0}^\infty\frac{z^n}{!_q(n)}
\end{equation}
where 
\begin{equation}
(a|q)_n=\prod_{k=0}^{n-1}(1-q^ka)\qquad (a|q)_\infty=\prod_{k=0}^\infty (1-q^ka).
\end{equation}
It is also worth recalling (\cite{SF}--Chapter~10, Section~2) the second Euler $q$-identity
\begin{equation}
\sum_{n=0}^\infty \frac{(-1)^nq^{\textstyle{\binom{n}{2}}}}{(q|q)_n} z^n=(z|q)_\infty
\end{equation}
valid for $|q|<1$ and all $z\in\Bbb{C}$. 

Let us now calculate the induced metric on $\Bbb{C}$. Applying the formula \eqref{K-induced} we obtain
\begin{equation}
\de s^2=\sum_{n=0}^\infty\frac{(1-q)q^n\:\de w \de\bar{w}}{(1+(1-q)q^n|w|^2)^2}. 
\end{equation}
With the help of a $q$-logarithm function 
\begin{equation}\label{q-log}
\begin{aligned}
\log_q(1+z)&=(1-q)\sum_{n\geq 0}
\frac{zq^n}{1+zq^n}\\
\log_q(1+z)&=\sum_{k\geq 0}\frac{(-1)^kz^{k+1}}{1+\cdots+q^k} 
\end{aligned}
\end{equation}
where the first formula is valid for arbitrary $z\in\Bbb{C}$ and the second in the unit disk $|z|<1$, 
the metric can also be expressed as 
\begin{equation}
\de s^2=\log_q'\bigl[1+(1-q)|w|^2\bigr]\:\de w\de\bar{w}.  
\end{equation}
We see that in the 
limit $q\rightsquigarrow 1^-$ this reproduces the classical Euclidean metric, and the first quantum correction 
looks elliptic
\begin{equation}
\de s^2\approx \bigl(1-2\frac{1-q}{1+q}|w|^2\bigr)\,\de w\de\bar{w}\qquad |w|^2\ll \frac{1}{1-q}. 
\end{equation}
We present more formulae related to these spaces in Appendix~B. 

\begin{prop} Modulo equivalence transformations of the quadratic form $Q$, the only algebras $\polq$ realizable 
as a Toeplitz operators algebra 
$\azz$ are precisely the Heisenberg-Weyl algebra and its above described $q$-variants. 
\end{prop}

\begin{proof} As we have just explained, the Heisenberg-Weyl algebra and $q$-variations are realizable as some $\azz$. To complete the proof, let us observe that the *-algebra 
$\azz$ can not possess characters, in accordance with Proposition~\ref{pointless}. This 
effectively excludes all other algebras $\polq$. 
\end{proof}

Let us now consider the Manin $q$-plane
\begin{equation}
z\bar{z}=q\bar{z}z\qquad 0<q<1. 
\end{equation}
The formula \eqref{form-general} for the quadratic form $\qform{}{}$ simplifies into
\begin{equation}\label{form-q-extended}
\qform{
\bar{z}^{i} z^{j}}{ 
\bar{z}^{n}z^{m}}= q^{in}s_{i+m}\, \delta_{i-j,n-m}
\end{equation}
a $q$-deformed version of \eqref{form-extended}. 

There is only one classical point here. 
It is the unique 
character evaluating to $0$ on $z$ and $\bar{z}$,  which refers to the center $0$ of the classical plane. 
The algebra therefore can not be realized within our principal framework.
But if we remove $0$ from 
consideration 
and allow Laurent series with singularity in $0$ as 
constituents of our Hilbert space $\cal{H}$, then a faithful representation is possible. 

Let us consider the space $\qinfty$ of generalized 
polynomials with all integer powers of $z$, and define the 
scalar product by requiring mutual orthogonality of the monomials $z^n$ and also 
\begin{equation}\label{zn-zn-binom}
\braket{z^n}{z^n}=q^{-\textstyle{\binom{n}{2}}} 
\end{equation}
for all $n\in\Bbb{Z}$. 
If we interpret $z$ as the multiplication operator by the coordinate $z$ and $\bar{z}$ as its formal 
adjoint, then 
the above non-commutation relation for the Manin $q$-plane is fulfilled, and we have a faithful representation. 
 
This space $\qinfty$ closes into a Hilbert space $\cal{J}$ of holomorphic functions over $\Bbb{C}-\{0\}$. Its reproducing kernel is given by 
\begin{equation}\label{K-triple-q}
K(\bar{w},z)=\sum_{n\in\Bbb{Z}}q^{\textstyle{\binom{n}{2}}}(\bar{w}z)^n.
\end{equation} 
Let us recall (\cite{SF}--Chapter 10, Section 4) the classical triple product identity
\begin{equation}
\sum_{n\in\Bbb{Z}}(-1)^nq^{\textstyle{\binom{n}{2}}}z^n=(z|q)_\infty(q/z|q)_\infty(q|q)_\infty
\end{equation}
which holds for $z\neq 0$ and $|q|<1$. Applying this to our reproducing kernel, we obtain 
\begin{equation}\label{classical-triple}
K(\bar{w},z)=(-\bar{w}z|q)_\infty(-q/\bar{w}z|q)_\infty(q|q)_\infty.
\end{equation}
The induced metric on $\Bbb{C}-\{0\}$ is calculated by applying \eqref{K-induced} to this infinite product. The first two symbols on the right side of  
\eqref{classical-triple} are 
transformed by the logarithm into infinite sums over $\Bbb{N}$, which after applying the derivation $\partial^2/\partial w\partial\bar{w}$ merge into a single sum over $\Bbb{Z}$.
The third symbol $(q|q)_\infty$ does not contribute to the metric, as it is a multiplicative constant transformed by $\log$ into an additive constant. 
Explicitly, we get 
\begin{equation}
\de s^2=\sum_{n\in\Bbb{Z}}\frac{q^n\,\de w\,\de\bar{w}}{(1+q^n\bar{w}w)^2}. 
\end{equation}

We see that the unique classical point represented by $0$ acts as a true geometrical singularity, as the metric diverges when $|w|$ tends to $ 0 $. 
On the other hand, far away from $0$
the metric becomes asymptotically Euclidean. 

Let $\Pi\colon\cal{J}\rightarrow\cal{J}$ be the orthogonal projector on the positive part, the Hilbert space $\cal{H}$ generated by the standard polynomials $\Cz$. 
Clearly $\Pi(\qinfty)=\cal{P}$ and in such a way, for the 
sequence $s_n$ given by the positive part of \eqref{zn-zn-binom} we obtain a non-commutative 
system satisfying Harmony Three. 
We can see this explicitly from  
$$ 
\frac{s_{n+1}}{s_n}=q^{-n}\qquad\quad \frac{s_{n+1}^2}{s_ns_{n+2}}=q. 
$$ 
This construction, which went `backwards' relative to our main considerations, ensures that Harmony Four holds, too. 
The space $\cal{H}$ consists of all entire functions of $\cal{J}$ with the 
reproducing kernel given by 
\begin{equation}
K(\bar{w},z)=\sum_{n\geq 0}q^{\textstyle{\binom{n}{2}}}(\bar{w}z)^n.
\end{equation} 
The induced metric is hence
\begin{equation}
\de s^2=\sum_{n\geq 0}\frac{q^n\,\de w\,\de\bar{w}}{(1+q^n\bar{w}w)^2}
=\log_q'(1+|w|^2)\frac{\de w\,\de\bar{w}}{1-q}. 
\end{equation}

The operators $Z$ and $Z^*$ satisfy the following non-commutation relation
\begin{equation}
ZZ^*=qZ^*Z -q\diada{0}{0}. 
\end{equation}
From this it is easy to see that the whole matrix algebra $\mathrm{M}_\infty(\Bbb{C})$ is included in $\azz$, and 
that the above relation is in fact a generating relation for $\azz$. The following short exact sequence holds: 
\begin{equation}
0\rightarrow \mathrm{M}_\infty(\Bbb{C})\longrightarrow \azz\longrightarrow \polq\rightarrow 0. 
\end{equation}
So, basically, this Toeplitz quantization smooths out the singularity at $0$. 
However, the change appears  `mild': The classical point at $0$ remains as the unique character of $\azz$. This is a counterexample for the possible 
removal of the generators condition in Proposition~\ref{pointless}. 

It is also interesting to observe that the constructed non-commutative system can be viewed as the Toeplitz quantization of the classical Euclidean plane in the sense that harmony property H4 
holds in this context, too.
This follows from the positivity analysis for the Stieltjes moment condition presented in Appendix~A 
together with this 
explicit calculation of the determinant generated by the full defining sequence: 
\begin{equation}\label{dete-binom}
\det \left |\begin{matrix} &  & \\
 & \smash[t]{q^{-{\textstyle\binom{i+j}{2}}}} & \\
 & & 
\end{matrix}\right |
=q^{-{\textstyle\binom{n}{2}\frac{4n-5}{3}}}\prod_{k=1}^{n-1}(1-q^{k})^{n-k}.
\end{equation}
Here the integer indexes $i,j$ run from $0$ to $n-1$.

A different Toeplitz quantization of the Manin $q$-plane is obtained if we choose the sequence \eqref{q-factorial}. Then, as explained in Appendix~B property H4 holds, too. It is interesting to observe that 
one and the same non-commutative system is viewed as a quantization of the classical plane as well as of the Manin $q$-plane. 

To complete our analysis of different types of quadratic algebras $\polq$, let us focus on the parabolic type 
given by \eqref{parabolic-z-z*}. Although the framework of entire functions does not work here, 
the algebra admits an interesting realization within a Hilbert space of holomorphic functions in 
the positive half-plane $\Re(z)>0$. Let us fix a positive number $h$. 
There exists a unique scalar product on the polynomial 
algebra  $\Cz$ such that the constant function $1=\ket{0}$ has unit norm and 
\begin{equation}
\zz\ket{0}=\frac{h}{2}\ket{0}.
\end{equation}
If we define (using the Euler gamma function $ \Gamma $)
the polynomials as  the shifted factorials 
\begin{equation}
\xi_n(z)=(z+\frac{h}{2})_n\qquad (z)_n:=\prod_{j=0}^{n-1}(z+j)=\frac{\Gamma(z+n)}{\Gamma(z)}, 
\end{equation}
then it is easy to see that 
\begin{equation}
\zz\xi_n(z)=(\frac{h}{2}+n)\xi_n(z)
\end{equation}
for every $n\in\Bbb{N}$. 
Clearly $\ket{0}=\xi_0(z)$ and these vectors form a basis of $\Cz$ with $\Czf{n}=
\mathrm{span}\{\xi_0(z),\xi_1(z),\dots,\xi_n(z)\}$  
Furthermore, a direct calculation reveals 
\begin{equation}
\braket{\xi_n(z)}{\xi_m(z)}=(h)_{n+m}. 
\end{equation}
Orthonormalization of these polynomials by the Gram-Schmidt procedure gives for the canonical orthonormal 
basis 
\begin{equation}
\ket{n}=\pp_n(z)=\prod_{j=0}^{n-1}\frac{z-j-h/2}{
\sqrt{(j+1)(h+j)}}=\frac{(-1)^n}{\sqrt{n! \,(h)_n}}(\frac{h}{2}-z)_n. 
\end{equation}
By using the Weierstrass product formula and Stirling asymptotic formula 
$$
\frac{1}{\Gamma(z)}=ze^{\gamma z}\prod_{n\geq 0}\Bigl\{(1+\frac{z}{n})e^{-z/n}\Bigr\}\qquad\qquad
\Gamma(x)\sim\sqrt{2\pi}x^{x-1/2}e^{-x}\quad x\rightsquigarrow +\infty
$$
it is easy to see that the kernel diagonal 
series is normally convergent on the half-plane $\Re(z)>0$. 
The polynomials $\pp_n(z)$ are also $\infty$-independent, so the conditions for Harmony One are here---although in this quite different 
context of the positive half-plane. 
Actually, the kernel can be explicitly summed to
\begin{equation}
K(\bar{w},z)=\frac{\Gamma(h)\Gamma(\bar{w}+z)}{\Gamma(\bar{w}+h/2)\Gamma(z+h/2)}.  
\end{equation}
This follows directly from the classical Gauss summation formula
$$
\sum_{n\geq 0}\frac{(a)_n(b)_n}{n!\, (c)_n}=\frac{\Gamma(c)\Gamma(c-a-b)}{\Gamma(c-a)\Gamma(c-b)}
$$
with $a=h/2-\bar{w}$, $b=h/2-z$ and $c=h$. 
The induced metric \eqref{K-induced} on the positive half-plane is especially simple, namely 
\begin{equation}
\de s^2=\de w\,\de\bar{w} \times \digamma'(\bar{w}+w)
\end{equation}
where $\digamma(z)=\Gamma'(z)/\Gamma(z)=[\log\Gamma(z)]'$ is the digamma function. 
Then taking into account the 
expansion around zero
$$ \digamma(z)=-\frac{1}{z}-\gamma-\sum_{n\geq 1}(-1)^n\zeta(n+1)z^n $$
as well as the classical positive half-plane hyperbolic metric $\de s^2={\de z\,\de\bar{z}}/{\Re(z)^2}$, 
we conclude that the quantum 
metric asymptotically behaves like the classical hyperbolic metric at the limit of the `infinity horizon' of space 
represented by imaginary numbers. 
Not surprisingly, as a matter of fact, this `geometrical heaven' is also interpretable as the classical part of this quantum
space as has been observed in our initial purely algebraic analysis. 

\section{Some Algebraic Aspects of Toeplitz Quantization}\label{algebraic}

In this section we shall 
describe an abstract algebraic underlying structure for the Toeplitz quantization. In particular,  
this will provide the foundations 
to fine-tune our principal construction of the extended space $\qinfty$ equipped with the 
quadratic form $\qform{}{}$ and the projection $\Pi$, to Hilbert spaces of entire functions and 
harmony conditions, where it is natural to use the more elaborate module $\W$, instead of polynomials.   

Let us assume that $\cal{P}$ is an everywhere 
dense linear subspace of a Hilbert space $\cal{H}$. Let $\cal{C}$ be a 
unital subalgebra of the *-algebra $M(\cal{P})$ of formally adjointable linear operators in $\cal{P}$. Then the conjugate algebra $\bar{\cal{C}}$ is also of the same category.  Let $\cal{A}\subseteq M(\cal{P})$ be the *-subalgebra generated by $\cal{C}$ and $\bar{\cal{C}}$. 

The space $\cal{P}$ can be viewed as a left $\cal{A}$-module, and in particular as both left $\cal{C}$-module as well as $\bar{\cal{C}}$-module. The algebra $\bar{\cal{C}}$ can be purely algebraically described as the opposite algebra of $\cal{C}$ equipped with the conjugate vector space structure. In this interpretation, the *-operation between 
$\cal{C}$ and $\bar{\cal{C}}$ is just the identity map. It naturally extends to *-involutions on the vector 
spaces $\cal{C}\otimes\bar{\cal{C}}$ and $\bar{\cal{C}}\otimes\cal{C}$ by 
\begin{equation}
(\alpha\otimes\bar{\beta})^*=\beta\otimes\bar{\alpha}\qquad\quad(\bar{\alpha}\otimes\beta)^*=\bar{\beta}\otimes\alpha. 
\end{equation}
In a similar way, the *-operation extends to arbitrary tensor products of $\cal{C}$ and $\bar{\cal{C}}$. In particular, 
it mixes $\cal{C}\otimes\cal{C}$ and $\bar{\cal{C}}\otimes\bar{\cal{C}}$. 
Let $m\colon\cal{C}\otimes\cal{C}\rightarrow\cal{C}$ and $\bar{m}\colon\bar{\cal{C}}\otimes\bar{\cal{C}}\rightarrow\bar{\cal{C}}$ be the corresponding product maps. Clearly, we have this commutative diagram: 
\begin{equation}
\begin{CD}
\cal{C}\otimes\cal{C}@>{\mbox{$m$}}>> \cal{C}\\
@V{\mbox{$*$}}VV @VV{\mbox{$*$}}V\\
\bar{\cal{C}}\otimes\bar{\cal{C}}@>>{\mbox{$\bar{m}$}}> \cal{C}
\end{CD}
\end{equation}
Let us assume that a linear map $\sigma\colon\cal{C}\otimes \bar{\cal{C}}\rightarrow \bar{\cal{C}}\otimes\cal{C}$ is 
given such that the diagram
\begin{equation}\label{sigma-*}
\begin{CD}
\cal{C}\otimes\bar{\cal{C}} @>{\mbox{$\sigma$}}>> \bar{\cal{C}}\otimes\cal{C}\\
@V{\mbox{$*$}}VV @VV{\mbox{$*$}}V\\ 
\cal{C}\otimes\bar{\cal{C}} @>>{\mbox{$\sigma$}}> \bar{\cal{C}}\otimes\cal{C}\\
\end{CD}
\end{equation}
is commutative. Furthermore, let us assume that the following pentagonal diagrams 
\begin{equation}
\qquad\quad
\begin{CD}
\cal{C}\otimes\cal{C}\otimes\bar{\cal{C}}  @>{\mbox{$m\otimes\id$}}>> \cal{C}\otimes\bar{\cal{C}}
@<{\mbox{$\id\otimes\bar{m}$}}<< \cal{C}\otimes\bar{\cal{C}}\otimes\bar{\cal{C}}\\
@V{\mbox{$(\sigma\otimes\id)(\id\otimes\sigma)$}}VV @VV{\mbox{$\sigma$}}V
@VV{\mbox{$(\id\otimes\sigma)(\sigma\otimes\id)$}}V\\
\bar{\cal{C}}\otimes\cal{C}\otimes\cal{C} @>>{\mbox{$\id\otimes m$}}> \bar{\cal{C}}\otimes\cal{C}
@<<{\mbox{$\bar{m}\otimes\id$}}< \bar{\cal{C}}\otimes\bar{\cal{C}}\otimes\cal{C}
\end{CD}
\end{equation}
are commutative too. 

\begin{remark} It is easy to see that the commutativity
of each of these pentagonal diagrams implies that 
of the other, if we assume the conjugational
symmetry \eqref{sigma-*}. 
The pentagonal diagrams as such, ensure that the space $\bar{\cal{C}}\otimes\cal{C}$
can be equipped with a natural associative product, such that
\begin{equation}
(\bar{\gamma}\otimes\alpha)(\bar{\beta}\otimes\delta)=\bar{\gamma}\sigma(\alpha\otimes\bar{\beta})\delta
\end{equation} 
and $1\otimes 1$ is the unit element of this algebra. This implies that always 
$\sigma(\alpha\otimes 1)=1\otimes\alpha$ and 
$\sigma(1\otimes\bar{\beta})=\bar{\beta}\otimes 1$.  
In particular $\sigma(1\otimes 1)=1\otimes 1$. 
The conjugational symmetry then is equivalent to the statement that the introduced $*$ is the *-structure on this algebra. 
Indeed, if 
$\sigma(\alpha\otimes\bar{\beta})=\sum_k\bar{\alpha_k}\otimes \beta_k$ then \eqref{sigma-*} is equivalent to 
$$
\sigma(\beta\otimes\bar{\alpha})=\sum_k\bar{\beta_k}\otimes\alpha_k\quad\Leftrightarrow\quad
[(1\otimes\alpha)(\bar{\beta}\otimes1)]^*=(1\otimes\beta)(\bar{\alpha}\otimes 1), 
$$
which in turn ensures that the *-operation on $\bar{\cal{C}}\otimes\cal{C}$ is anti-multiplicative. Let us denote by 
$\cal{B}$ the resulting *-algebra. By construction both $\cal{C}$ and $\bar{\cal{C}}$ are subalgebras of $\cal{B}$, and they generate $\cal{B}$ with the commutation rule given by $\sigma$. 
\end{remark}

Let us now define an extended representation space
\begin{equation}
\pinfty=\cal{B}\otimes_{\cal{C}}\cal{P}\leftrightarrow\bar{\cal{C}}\otimes\cal{P}. 
\end{equation}
Because of the above identification, the space $\pinfty$ is a left $\cal{B}$-module, in a natural way: the module structure is induced simply by left multiplication. 
By construction, the space $\cal{P}$ is naturally viewable 
as a $\cal{C}$-submodule of $\pinfty$. 
The inclusion map is given by $\cal{P}\ni\psi
\longmapsto 1\otimes\psi\in\pinfty$. 
So $\cal{P}$ is cyclic for the $\cal{B}$-module $\pinfty$. 

The left $\bar{\cal{C}}$-module structure on $\cal{P}$ is naturally expressed via the projection map $\Pi\colon
\pinfty\rightarrow\pinfty$ given by
\begin{equation}
\Pi(\bar{\alpha}\otimes \psi)=1\otimes (\bar{\alpha}\psi). 
\end{equation} 
Clearly, the image of the idempotent $\Pi$ is $\cal{P}$ and $\Pi$ is $\bar{\cal{C}}$-linear. 

\begin{lemma} The formula
\begin{equation}
\qform{\bar{\alpha}\otimes\varphi}{\bar{\beta}\otimes\psi}
=\sum_k\langle\beta_k\varphi,\alpha_k\psi\rangle
\end{equation}
where $\sigma(\alpha\otimes\bar{\beta})=\sum_k\bar{\beta}_k\otimes\alpha_k$, 
defines a sesqui-linear form on $\pinfty$ which restricted on $\cal{P}$ reproduces the original scalar product. With respect of this form, the left $\cal{B}$-module structure is symmetric. 
\end{lemma}
\begin{proof} Remember that $ \sigma(1\otimes 1)=1\otimes 1$ so $\qform{1\otimes\varphi}{1\otimes\psi}
=\langle\varphi,\psi\rangle$. 
Furthermore since $\sigma(\beta\otimes\bar{\alpha})=\sum_k\bar{\alpha}_k\otimes\beta_k$ so 
$$ \qform{\bar{\beta}\otimes\psi}{\bar{\alpha}\otimes \varphi}=
\sum_k\langle\alpha_k\psi,\beta_k\varphi\rangle=\bigl[\sum_k\langle\beta_k\varphi,\alpha_k\psi\rangle\bigr]^*
=\qform{\bar{\alpha}\otimes\varphi}{\bar{\beta}\otimes\psi}^* $$
and hence the hermitian symmetry of $\qform{}{}$. In order to prove that the action of $\cal{B}$ is symmetric, it is sufficient to check it for $\cal{C}$, or equivalently $\bar{\cal{C}}$. If $c\in\cal{C}$ then
$$
\qform{\bar{c}\bar{\alpha}\otimes \varphi}{\bar{\beta}\otimes\psi}=
\sum_k \langle u_k\varphi,v_k\psi\rangle=\qform{\bar{\alpha}\otimes\varphi}{c(\bar{\beta}\otimes\psi)}
\qquad\sum_k \bar{u}_k\otimes v_k=\sigma(\alpha c\otimes \bar{\beta})
$$
and we have applied the pentagonal symmetry for $\sigma$ in twisting the product 
$\alpha c$ with $\bar{\beta}$. 
\end{proof}

\begin{lemma} The projection $\Pi$ is symmetric relative to $\qform{}{}$. 
\end{lemma}
\begin{proof} Remembering that $\sigma$ classically flips $1$ with everything, we obtain 
\begin{align*}
\qform{\Pi(\bar{\alpha}\otimes\varphi)}{\bar{\beta}\otimes\psi}
&=\qform{1\otimes \bar{\alpha}\varphi}{\bar{\beta}\otimes\psi}=
\braket{\beta\bar{\alpha}\varphi}{\psi} 
\\&=
\braket{\bar{\alpha}\varphi}{\bar{\beta}\psi}=
\qform{\bar{\alpha}\otimes\varphi}{\Pi(\bar{\beta}\otimes\psi)}. 
\end{align*}
So $\Pi$ is indeed symmetric. 
\end{proof}

Let us now consider the null space of $\qform{}{}$. In general, it will be a non-trivial subspace of $\pinfty$. It is defined by 
$$
\ninfty=\Bigl\{\eta\in\pinfty \bigm\vert \qform{\eta}{\pinfty}=\{0\}\Bigr\}.
$$
\begin{lemma} The space $\ninfty$ is $\cal{B}$-invariant and moreover 
$$\Pi(\ninfty)=\{0\}=\cal{P}\cap\ninfty. $$
\end{lemma}
\begin{proof} The invariance property is a direct consequence of the symmetry of the action of $\cal{B}$ with respect to $\qform{}{}$.
Since $\Pi$ is symmetric, it also follows that $\ninfty$ is $\Pi$-invariant. 
On the other hand, 
$\Pi$ projects onto $\cal{P}$, and there $\qform{}{}$ reduces to the initial strictly positive scalar product.  
\end{proof}

Consequently, the whole structure naturally projects on the factor module
\begin{equation}
\qinfty=\pinfty/\ninfty 
\end{equation}
over the algebra $\cal{B}$. 
We shall use the same symbols $\qform{}{}\colon\qinfty\times\qinfty\rightarrow \Bbb{C}$ and $\Pi\colon\qinfty\rightarrow\qinfty$ for the 
projected objects. Clearly, the projected $\qform{}{}$ is non-degenerate and $\cal{P}$ is a $\cal{C}$-submodule 
of $\qinfty$. 
Because of the symmetry of $\Pi$ we can write 
\begin{equation}
\qinfty=\cal{P}\oplus\cal{P}^\bot, 
\end{equation}
and $\cal{P}^\bot$ is also a $\cal{C}$-submodule of $\qinfty$. To every element of $\cal{B}$ we can associate a 
linear operator from $\cal{A}$ via the symbol construction: 
\begin{equation}
\cal{B}\ni b\longmapsto \Pi b\Pi\colon \cal{P}\rightarrow\cal{P}. 
\end{equation}
Explicitly, if 
$b=\sum_k\bar{\alpha}_k\otimes \beta_k$,  
then simply 
\begin{equation}\label{simply-compress}
\Pi b\Pi\leftrightarrow \sum_k\bar{\alpha}_k\beta_k,
\end{equation}
and in particular this is an element of $\cal{A}$. 
\begin{remark}
These elements in \eqref{simply-compress} 
span the $\bar{\cal{C}}$--$\cal{C}$ bimodule in $\cal{A}$ 
generated by $1\in\cal{A}$. 
In general, this will be strictly 
smaller than $\cal{A}$. 
\end{remark}

The construction described here 
can be reversed, thereby 
giving its own abstract characterization, 
which is very much in resonance 
with the Stinespring construction (see \cite{Sti}), 
which itself generalizes the GNS construction to the level of
appropriate completely 
positive maps. 
We can start from the twisted product algebra $\cal{B}\leftrightarrow \bar{\cal{C}}\otimes\cal{C}$
as above, represented by formally adjointable operators in a linear space $\qinfty$ equipped with a 
not necessarily positive regular scalar product $\qform{}{}$. 

Let $\cal{P}$ be a linear subspace of $\qinfty$ satisfying the following four structural properties:
\begin{itemize}
\item It is {\it positive}, in the sense that $\qform{}{}$ reduces to a strictly positive scalar product on $\cal{P}$; 
\item It is {\it orthocomplementable} in $\qinfty$ in the sense that 
$$ \qinfty=\cal{P}\oplus\cal{P}^\bot $$ 
with respect to $\qform{}{}$; 
\item It is $\cal{C}$-invariant; 
\item It is {\it cyclic}, in the sense that 
$$ \bar{\cal{C}}\cal{P}=\Bigl\{\sum\bar{c}\psi\,\vert \,c\in\cal{C},\psi\in\cal{P}\Bigr\}=\qinfty. $$
\end{itemize}
\begin{prop}
Under the above assumptions, the entire system is naturally isomorphic to the one constructed above. \qed
\end{prop}

\section{Concluding Thoughts}\label{conclusions}

Our Toeplitz quantization construction provides a kind of resonant bridge
going between two, in principle very different, algebras. 
On one side we have a given polynomial quadratic *-algebra $\polq$ generated by abstract 
symbols $z$ and $\bar{z}$. 
On the other there is a concrete 
non-commutative *-algebra of operators $\azz$, acting within a Hilbert space $\cal{H}$ of 
entire functions, possessing a dense common invariant subspace $\cal{W}$, 
and also viewable as transformations of the whole $\HolC$, generated by the multiplication operator $z$ and its companion $\zzinfty$, interpretable as the adjoint of $z$ in terms of $\cal{H}$.   
 
There is an extended space 
$\qinfty$, which is equipped 
with a non-necessarily positive non-degenerate scalar product $\qform{}{}$, 
on which the *-algebra $\polq$  symmetrically acts. The common domain $\cal{W}$ 
for the operators of $\azz$ is isometrically realized as an orthocomplementable subspace of $\qinfty$. 
Then $\polq$ is morphed into $\azz$ with the help of the orthogonal projection $\Pi\colon\qinfty\rightarrow 
\cal{W}$. In this sense $\azz$ is interpreted as a Toeplitz quantization of $\polq$. 

At a first sight, it might come as a surprise that we can perform the construction on two a priori unrelated 
algebras such as $\polq$ and $\azz$---coming from two apparently quite different worlds. 

On one hand, we believe 
that this reflects the flexibility and versatility of the construction, which as we have seen admits an elegant purely algebraic foundation. 
And after all, behind both algebras lies an 
intuitive idea of a quantum object resembling the classical Euclidean plane. On the other hand, it is natural 
to look for some additional geometrical context, in which the relationship between $\polq$ and $\azz$ looks especially harmonic.  

One such a context is given by our last harmony condition, which 
requires that $\qform{}{}$ be positive.
As we have already mentioned, in the case of the standard commutative product in 
$\Bbb{C}[z,\bar{z}]$ the positivity condition leads to its {\it moment problem} of finding a finite measure on $\Bbb{C}$ reproducing the scalar product. Such a measure can always be taken to be radial for diagonal scalar products. 

Our principal series of examples, as in \cite{coherent}, also comes from a diagonal scalar product on $\Cz$, defined by a positive sequence $s_n=s_{nn}$ of squares of norms of mutually orthogonal $z^n$. There is a natural generalization, which includes as a special case the parabolic quadratic algebra considered in Section~2, and 
which is also invariant under Euclidean transformations of $\Bbb{C}$. 
If we require that the polynomials of the canonical orthonormal basis satisfy 
\begin{equation*}
\pp_n(z)\mid \pp_{n+1}(z)
\end{equation*}
for every $n\in\Bbb{N}$, that is to say, that zeroes of $\pp_n(z)$ are included (multiplicities counted) in zeroes of $\pp_{n+1}(z)$,  the whole scalar product is encoded in two sequences.
these are a sequence of positive numbers 
$s_n$ and a sequence of complex numbers $\lambda_n$ so that   
\begin{equation*}
\pp_n(z)=\frac{1}{\sqrt{\smash[b]{s_n}}}(z-\lambda_1)\cdots(z-\lambda_n)
\end{equation*}
or equivalently 
\begin{equation*}
z^n=\sum_{k=0}^n\sqrt{s_k}u_{n-k}(\lambda_1,\cdots,\lambda_{k+1})\pp_k(z)
\end{equation*}
where $u_m$ are symmetric functions computing the sum of all elementary monomials of degree $m$ 
constructed from their arguments (for instance $u_0$  is always $1$ while $u_1$ sums
its arguments and $u_7(a,b)=(a+b)(a^2+b^2)(a^4+b^4)$).  

The diagonal scenario is recovered as the special case $\lambda_n=0$ for all $n$, and the parabolic quadratic algebra example corresponds to $\lambda_n=n-1+h/2$. The algebraic part of Harmony Zero will 
be always satisfied, as $Z^*\Czf{n}\subseteq\Czf{n}$ for every $n\in\Bbb{N}$. 

We have analyzed in this paper the `one-particle' scenario, where there is only one complex variable $z$ to quantize. In a very similar way, everything is extendible to several complex variables. The methods should 
also be extendible to classical subdomains in $\Bbb{C}^n$. Here we would expect a natural emergence of 
non-Euclidean quantum geometries in resonance with the results of \cite{MP}, where quantum hyperbolic 
planes are studied, and the interpretation of Poincar\'e models 
linked with the Hilbert spaces of analytic functions on the open unitary disk $\Bbb{D}$.

Another very interesting and promising topic for future 
explorations is to consider 
multi-dimensional quantum Euclidean spaces defined via the described quantization, which 
preserve some important symmetries and develop differential geometry on such spaces. It is tempting to 
procede in the spirit of \cite{qdunkl, qdunkl2}, where it was revealed that classical structures with discrete symmetries exhibit a `hidden' quantum personality, 
which is describable in the elegant geometrical language of quantum principal bundles. 

\appendix

\section{Positivity and Stieltjes Moment Condition} 

In this Appendix we shall discuss in more detail the positivity property for the quadratic form $\qform{}{}$, induced 
on the classical polynomial *-algebra $\Bbb{C}[z,\bar{z}]$ by a diagonal scalar product $\braket{}{}$ on $\Cz$ with 
the defining sequence $s_n$ given in 
\eqref{simple-sequence}

The geometry of $\Bbb{C}[z,\bar{z}]$ with such a $\qform{}{}$ is reflected in the accompanying figure. We start with an integer lattice. The white nodes correspond to the monomials $\bar{z}^k z^l$ with 
$k,l\geq 0$. These monomials are a basis in the commutative *-algebra $\Bbb{C}[z,\bar{z}]$. By allowing 
negative powers of $z$ and $\bar{z}$ such that $k+l\geq 0$, this algebra is included 
in an extended commutative *-algebra $\Bbb{E}[z,\bar{z}]$. 
Such monomials are represented by 
grey nodes. 
To each node $ a = (k,l) $ we define
its {\em value} to be the number $|a|:= k+l$. 
In terms of this picture the scalar product $\qform{}{}$ is computed using \eqref{form-extended} 
according to the following rule: if the nodes labeled by $a$ and $b$ 
share one of the parallel diagonal lines, 
then the scalar product is $s_{(|a|+|b|)/2}$. Otherwise the monomials are orthogonal. 
This
immediately reveals that $\qform{}{}$ will be positive if and only if its restriction to all of the diagonal 
subspaces is positive. 

\begin{figure}
\begin{center}
\includegraphics[scale=0.25]{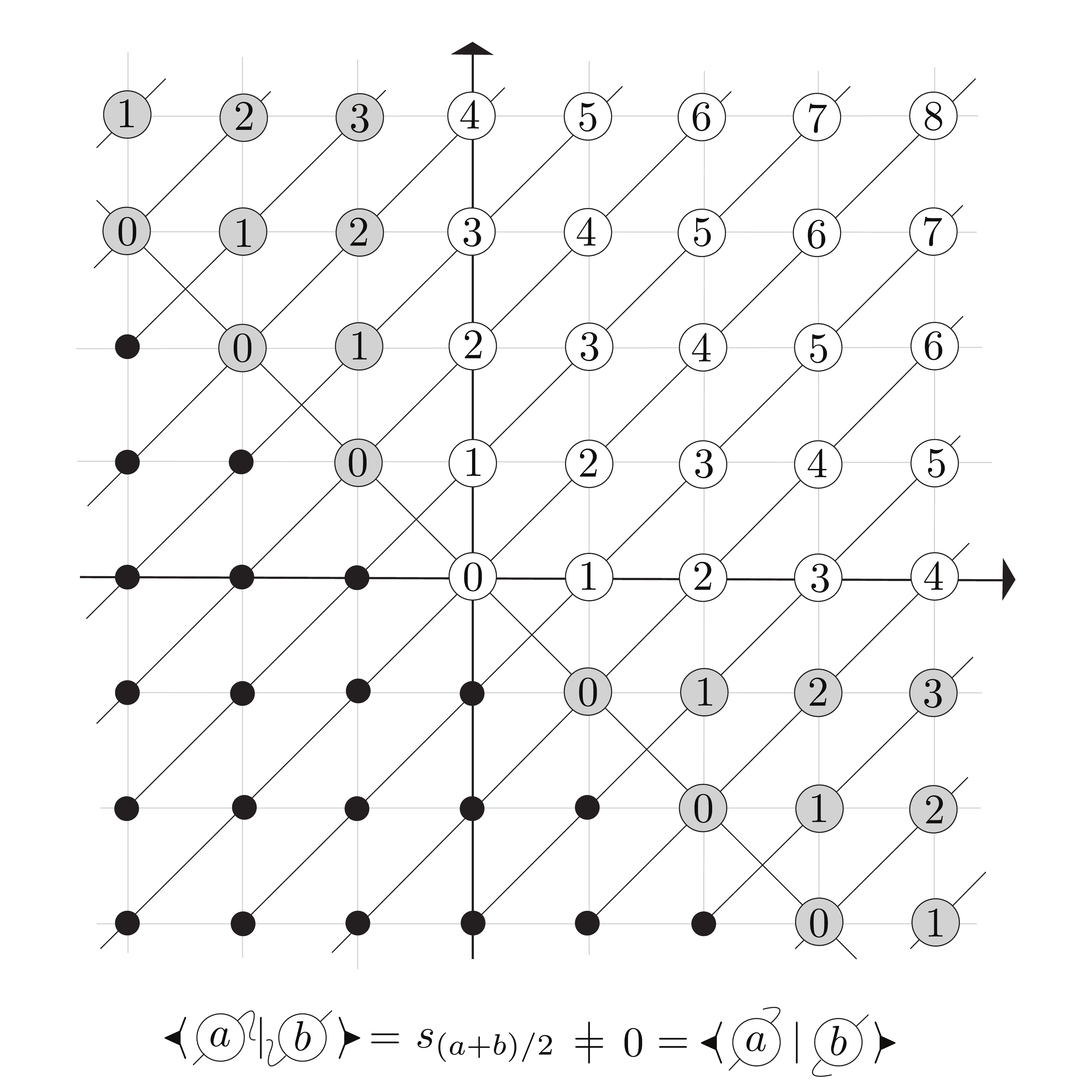}
\end{center}
\end{figure}

But the diagonal subspaces are of 
two canonical types: those consisting of 
nodes with even value and those consisting of nodes
with odd value. 
So the positivity condition reduces to the positivity of the corresponding Gram matrices, 
$$
\begin{pmatrix}
s_0 & s_1 & \cdots & s_n \\
s_1 & s_2 & \cdots & s_{n+1}\\
\vdots & \vdots &\ddots & \vdots\\
s_n & s_{n+1} & \cdots & s_{2n} 
\end{pmatrix}
\qquad\quad
\begin{pmatrix}
s_1 & s_2 & \cdots & s_n \\
s_2 & s_3 & \cdots & s_{n+1}\\
\vdots &\vdots &\ddots & \vdots\\
s_n & s_{n+1} & \cdots & s_{2n-1} 
\end{pmatrix}
$$
of even and odd types respectively. But this is exactly the Stieltjes moment condition for the existence of a 
measure $\mu$ on $[0,\infty)$ satisfying 
$$ 
s_n=\int_0^\infty t^n\, d\mu(t) \qquad\quad\forall n\in\Bbb{N}. 
$$
And this is equivalent to the existence of 
a rotationally invariant measure $d\mu(z,\bar{z})$ 
on $\Bbb{C}$ reproducing the scalar product on $\Bbb{C}[z,\bar{z}]$ in the standard way: 
$$
\qform{f}{g}=\int_{\Bbb{C}} \overline{f(z)}g(z)\,d\mu(z,\bar{z}). 
$$  

Interestingly, we can look at all of this another way around and 
use the construction to prove the Stieltjes moment condition. Indeed, as is easily seen from 
the figure, the positivity of $\qform{}{}$ on $\Bbb{C}[z,\bar{z}]$ is equivalent to the positivity of 
$\qform{}{}$ on $\Bbb{E}[z,\bar{z}]$. And such an 
extended positivity is equivalent, according to the theorem of Stochel–Szafraniec \cite{schmudgen},
to the existence of an underlying measure for the quadratic form. 
This measure, due to the rotational symmetry of the system, can be 
always taken to be purely radial. 

\section{On Hilbert Spaces of Holomorphic Functions}
\subsection*{Reproducing Kernels as Point Functionals}

We shall here review some general properties of Hilbert spaces of holomorphic functions, which are particularly 
relevant for our main considerations on Toeplitz quantization. 

We refer to \cite{pick} for a lovely self-contained exposition of the theory of Hilbert function spaces, including those consisting of holomorphic functions.
In dealing with these structures, it is optimal to assume that everything occurs in a mathematical universe in which linear functionals defined over Hilbert spaces are always continuous. 
This leads to interesting foundational issues in Functional Analysis, involving the axiom of choice, automatic continuity and constructibility of objects. We refer to \cite{Schechter}, 
especially Chapter~6 of the book, for an in-depth discussion of these conceptual roots.  

Let $\cal{H}$ be a Hilbert space consisting of holomorphic functions over a domain $\Omega$.
So $\cal{H}$ is a complex vector subspace of $\mathrm{Hol}(\Omega)$. The structure of such a space is completely 
determined by its reproducing kernel function $K\colon\Omega\times\Omega\rightarrow \Bbb{C}$, which in 
the main text we have 
met in the context  $\Omega=\Bbb{C}$ of entire functions. 
The primary properties we mentioned in that special 
context are valid for a general domains. 
In particular, the kernel is holomorphic 
in its second argument and anti-holomorphic in its first argument. 
And for every $w\in\Omega$ we can define 
$[w]=w(z)=K(\bar{w},z)$ to be the corresponding `point wave function'. These functions belong to 
$\cal{H}$ and reproduce the values of the functions from $\cal{H}$ in the points $w$ via the scalar product 
$\langle w(z),\psi(z)\rangle=\psi(w)$.  
So the scalar product between two point wave functions 
$ [v], [w] $
is the reproducing kernel
of the points $ v,w \in \Omega $ taken in the opposite order $K(\bar{w},v)=\braket{v(z)}{w(z)}=[w](v)=
\overline{K(\bar{v},w)}$. 

If $\Lambda$ is any subset of $\Omega$ possessing an accumulation point in $\Omega$, then the point
wave functions $\lambda(z)$ with $\lambda\in\Lambda$ generate the whole Hilbert space $\cal{H}$ or, in other words, the orthocomplement to all of them is $\{0\}$. 

And if $\Lambda$ is a finite non-empty subset of $\Omega$ we can construct a quadratic matrix $K[\Lambda]$
whose entry $(v,w)$ is $K(\bar{w},v)$, where now $w,v\in\Lambda$. 
Such a matrix is always positive. It is the Gram matrix for the 
system of vectors $\lambda(z)$ with $\lambda\in\Lambda$. 
It turns out that conversely, any bi-anti-holomorphic function 
$K$ on $\Omega\times\Omega$ satisfying the positivity condition for its all matrices $K[\Lambda]$
generates a Hilbert space of holomorphic functions. 

We can naturally extend the notion of wave point function to include derivative operators. 
For every $w\in\Omega$ and $n\in\Bbb{N}$ let us define a function $\pointf{w}{n}=\pointf{w}{n}(z)\in \cal{H}$ by 
\begin{equation}
\braket{\pointf{w}{n}}{\psi}=\frac{\partial^n\psi}{\partial z^n}(w) \qquad \forall \psi \in \cal{H}.  
\end{equation}
This makes sense since the right side is a 
norm continuous functional of $ \psi \in \cal{H}$. 
So we have included the original point wave functions as $\pointf{w}{0}=[w]$. In terms of the reproducing kernel we have 
\begin{equation}
\pointf{w}{n}(z)=\frac{\partial^n \! K}{\partial \bar{w}^n}(\bar{w},z) 
\end{equation}
and the classical Cauchy formula can be expressed as  
\begin{equation}
\pointf{w}{n}=\frac{n!}{2\pi i}\oint \frac{d c\,  [c]}{(c-w)^{n+1}}, 
\end{equation}
where the integral is over a positively oriented simple curve in $\Omega$ around $w$. From the definition of these higher order point wave functions 
it follows that 
\begin{equation}
\braket{\pointf{c}{m}}{\pointf{w}{n}}=\frac{\partial^{n+m}}{\partial \bar{w}^n\partial c^m}K(\bar{w},c).
\end{equation}

\begin{lemma} For a given $w\in\Omega$ the functions $\bigl\{\pointf{w}{n}: n\in\Bbb{N}\bigr\}$ span a dense linear subspace in $\cal{H}$. 
\end{lemma}

\begin{proof} For a function from $\cal{H}$, being orthogonal to all $\pointf{w}{n}$ with fixed $w$ means vanishing at $w$, together with all its derivatives. 
Because of the connectedness of $\Omega$, such a holomorphic function must be zero identically. 
\end{proof}

\begin{lemma} 
For a given $w\in\Omega$ the functions $\bigl\{\pointf{w}{n}: n\in\Bbb{N}\bigr\}$ admit a biorthogonal system if and only if $\Cz\subset \cal{H}$. In this case the corresponding 
biorthogonal system is given by $\bigl\{(z-w)^n/n!: n\in\Bbb{N}\bigr\}$. 
\end{lemma}

\begin{proof} If $b_m(z)$ are functions biorthogonal to $\pointf{w}{n}$ then $$\braket{\pointf{w}{n}}{b_m}=\delta_{nm}
=\frac{\partial^n}{\partial z^n}b_m(z)\Bigm\vert_{z=w} $$
and by expanding in Taylor series around $w$ we obtain $b_m(z)=(z-w)^m/m!$. In particular $\Cz\subset\cal{H}$. Conversely, if polynomials are included in $\cal{H}$ then $b_m(z)$ defined by the same formula are clearly the 
biorthogonal system for $\pointf{w}{n}(z)$. 
\end{proof}

We see that the inclusion of the polynomials $\Cz$ in $\cal{H}$ is a natural structural property of $\cal{H}$. However, 
the space $\Cz$ is not necessarily dense in $\cal{H}$. We can proceed in the same spirit here, as 
in the context of entire functions, and 
consider the array of harmony properties, related to the operator form of the complex 
coordinate $z$, and its adjoint operator, introduce 
the common invariant subspace $\cal{W}$ consisting of normal vectors, 
and obtain a non-commutative *-algebra $\azz$ 
which itself can be viewed as a form of quantization of the domain $\Omega$, and also as the algebra 
of the Toeplitz operators for the Toeplitz quantization of appropriate quadratic algebras. Here it is natural to assume the maximality of $\Omega$ relative to $\cal{H}$ in the sense that there is no common analytic extension beyond $\Omega$, valid for all the elements of $\cal{H}$. Of course, when $\Omega=\Bbb{C}$ this condition is trivial. 

\subsection*{Segal-Bargmann Space and $q$-Deformations}

Let us first consider the Hilbert space $L^2(\Bbb{C},\rho)$ where $\rho(z)=\exp(-\vert z\vert^2)/\pi$. 
The Segal-Bargmann space then consists of all entire functions belonging to the above Hilbert space. 
The polynomials in $ z $
are everywhere dense, and the monomials in $ z $ are mutually orthogonal. The reproducing kernel 
is given by
\begin{equation}
K(\bar{w},z)=\exp(\bar{w}z).
\end{equation}
The defining sequence of weights 
(see \eqref{simple-sequence}) is 
\begin{equation}
s_n=s_{nn}=n!
\end{equation}
as is also immediately visible from the expansion of the exponent in the reproducing kernel formula. 
Modulo normalizations 
this scalar product is the unique scalar product on $\cal{P} = \Cz$ with respect to which 
$Z^*=\partial/\partial z$. 

Let us now consider a $q$-deformation of the Segal-Bargmann space 
with the defining sequence given 
by the $q$-factorials \eqref{q-factorial}. As we have seen, the reproducing kernel is given by the $q$-exponential 
function $K(\bar{w},z)=E_q(\bar{w}z)$. There are infinitely many generating measures on $\Bbb{C}$ for this reproducing 
kernel, which in particular means that $\qform{}{}$ is strictly positive on $\Bbb{C}[z,\bar{z}]$. 

Here are two distinguished measures, both rotationally symmetric. The first one is discrete on the radial coordinate 
in $\Bbb{C}$ and based on the discrete $q$-integral 
\begin{equation}
\int_0^\infty f(t)\, d_q t=\sum_{n\in\Bbb{Z}}q^nf(q^n). 
\end{equation}
A straightforward application of the $q$-integral formula 
\begin{equation}\label{q-factorial-id}
\int_0^\infty \frac{t^{\alpha}\, d_qt}{E_q(t/q)}=\frac{(q|q)_\alpha}{(q-1|q)_\alpha(q/(q-1)|q)_{-\alpha}}=:!_q(\alpha)
\end{equation}
together with the observation that $!_q(\alpha)=s_\alpha$ for $\alpha\in\Bbb{N}$, 
leads us to this radially discrete density function 
\begin{equation}\label{rho-discrete}
\rho(z,\bar{z})=\frac{1}{\pi}\frac{1}{E_q(|z|^2/q)}\sum_{n\in\Bbb{Z}} q^n\delta(|z|^2-q^n).
\end{equation}
We have used an extended definition for the $q$-symbols
\begin{equation}
(z|q)_\alpha=\frac{(z|q)_\infty}{(zq^\alpha|q)_\infty}
\end{equation}
valid for arbitrary complex numbers $\alpha\in\Bbb{C}$. 

On the other hand, we can use the continuous variation 
\begin{equation}\label{q-factorial-ic}
\int_0^\infty \frac{t^{\alpha-1}\,dt}{(-t|q)_\infty}=(q^{1-\alpha}|q)_\alpha\frac{\pi}{\sin(\pi\alpha)}
\end{equation}
of the integral \eqref{q-factorial-id}. This for $\alpha\in\Bbb{N}$ morphs into
\begin{equation}
\int_0^\infty \frac{t^{\alpha}\,dt}{E_q(t/q)}=\frac{\log(1/q)}{1/q-1}s_\alpha 
\end{equation}
and we naturally arrive to 
another density function, scaled by a positive factor and free of discrete $\delta$-terms
\begin{equation}
\rho(z,\bar{z})=\frac{1/q-1}{\pi\log(1/q)}\frac{1}{E_q(|z|^2/q)}.
\end{equation}
In the limit $q\rightsquigarrow 1^-$ both measures become the standard Segal-Bargmann measure. 

It is interesting to observe that from the sequence of weights we can see directly that the form 
$\qform{}{}$ is strictly positive. Indeed, 
consider the following $n\times n$ matrix determinant, 
with entries indexed by $i,j=0,\dots,n-1$ and composed of $q$-factorials: 
\begin{equation}\label{qf-dete}
\left |
\begin{matrix} &  & \\
 & !_q(i+j+d) & \\
 & & 
\end{matrix} \right |
=q^{-{\textstyle \binom{n}{2}}(d+{\textstyle \frac{2n-1}{3}})}\prod_{k=0}^{d-1}\frac{!_q(n+k)}{!_q(k)}
\Bigl\{\prod_{k=1}^{n-1} !_q(k)\Bigr\}^2 
\end{equation}
where $d\in\Bbb{N}$. For a positive $q$ these are all positive numbers (and quickly growing very large if 
$q\leq 1$). Here is the sequence of positive numbers whose partial products generate the above determinants
\begin{equation}
\lambda_n=q^{-n^2}!_q(n)^2\qquad\lambda_n=q^{-n^2-nd}!_q(n+d)\,!_q(n)^2
\end{equation}
corresponding to $d=0$ and $d\geq 1$ respectively. 
From the analysis of the previous Appendix it follows that the form $\qform{}{}$ must be strictly positive on $\Bbb{C}[z,\bar{z}]$. 

And here is a variation on the above sequences
\begin{equation}
\lambda_n=q^{n^2}!_q(n)^2\qquad\qquad\:\lambda_n=q^{n^2}!_q(n+d)\,!_q(n)^2
\end{equation}  
generating the determinants
\begin{equation}\label{qf-dete-2}
\left |\begin{matrix} &  & \\
 & q^{ij}!_q(i+j+d) & \\
 & & 
\end{matrix}\right |
=q^{{\textstyle \binom{n}{2}}{\textstyle \frac{2n-1}{3}}}\prod_{k=0}^{d-1}\frac{!_q(n+k)}{!_q(k)}
\Bigl\{\prod_{k=1}^{n-1} !_q(k)\Bigr\}^2
\end{equation}
via their partial products. The positivity of these numbers can 
be used to derive the positivity of $\qform{}{}$ on the space 
$\polq$ for the Manin $q$-plane, so in this scenario the property H4 holds, too.
Note that the parameter $q$ is the same for both algebras. 

\subsection*{Bergman Spaces}
These spaces pro\textbf{}vide one of the earliest frameworks for quantizing Euclidean domains, via complex functions theory. They were introduced by Stefan Bergman \cite{berg} during the 20s of the 20th century. 

Let us consider a bounded domain $\Omega$ equipped with its standard Euclidean measure. Let $\cal{H}$ be the space of all 
square integrable holomorphic functions in $\Omega$. It is easy to see that $\cal{H}$ is a closed subspace of 
$L^2(\Omega)$ and thus, in the induced scalar product, it becomes a Hilbert space of holomorphic functions. 

The scalar product is given by 
 $$ \langle f,g\rangle=\int_\Omega \overline{f(w)}g(w). $$
Because of the boundedness of $\Omega$, the multiplication operator by $z$ is bounded and defined on all 
$\cal{H}$. Together with its adjoint, it generates a non-commutative C*-algebra representing a quantized domain 
$\Omega$. 

We then have 
$$
\int_\Omega \vert w(z)\rangle\langle w(z)\vert=1\colon\cal{H}\rightarrow\cal{H} 
$$
where the integral is in the weak operator topology. This is a direct consequence of the definition of the space $\cal{H}$. It is worth noticing that in the special case of $\Omega=\Bbb{D}$ the unit disk, the
reproducing kernel is given by 
$$
K(\bar{w},z)=\frac{1}{\pi}\frac{1}{(1-\bar{w}z)^2}
$$
and the whole system is interpretable as a quantization of the Poincar\'e model of the hyperbolic plane. 

\subsection*{The Toeplitz Extension}
Another quantum model of the hyperbolic plane is obtained by considering the square summable power series 
\cite{debranges2}. The unit disc $\Bbb{D}$ as the common domain of holomorphicity, and the reproducing 
kernel is simply
$$ 
K(\bar{w},z)=\frac{1}{1-\bar{w}z}
$$ 
so that the monomials $z^n$ form an orthonormal basis in $\cal{H}$. The coordinate $z$ acts as the unilateral 
shift operator in $\cal{H}$. 

\subsection*{Paley-Wiener $\&$ Euler Spaces}
The Paley-Wiener space of index $a>0$ is generated by the reproducing kernel
$$
K(\bar{w},z)=\frac{\sin[a(\bar{w}-z)]}{\pi(\bar{w}-z)}. 
$$
The elements of the space are the entire functions square-integrable over the reals, and the scalar product 
is given by 
$$
\langle f(z),g(z)\rangle=\int_{-\infty}^{\infty}\overline{f(t)}g(t)\, dt. 
$$
The space can be viewed as the image of $L^2[-a,a]$ via the complexified Fourier transform. More generally, 
if we consider an entire function $E(z)$ satisfying 
$$
\vert E(x-iy)\vert < \vert E(x+iy)\vert 
$$
where $x\in\Bbb{R}$ and $y>0$, then the Hilbert space $\cal{H}(E)$ is given by the reproducing kernel
$$
K(\bar{w},z)=\frac{E^*(\bar{w})E(z)-E(\bar{w})E^*(z)}{2\pi i(\bar{w}-z)}. 
$$ 
The space $\cal{H}(E)$ consists of entire functions. They are all square integrable over the reals with the weight function $1/\vert E(t)\vert^2$ so that the scalar product is given by 
$$
\langle f(z),g(z)\rangle=\int_{-\infty}^{\infty}\frac{\overline{f(t)}g(t)}{\vert E(t)\vert^2}\, dt. 
$$
The Paley-Wiener spaces are a special case when $E(z)=\exp(-iaz)$.
We refer to \cite{debranges} for a detailed study of these `Euler spaces'. They are a rich collection of examples 
of the Hilbert spaces of entire functions going beyond our harmony properties. Nevertheless, they allow 
similar non-commutative constructions and provide, with the geometrical interpretations
appropriately refined, an interesting complementary
framework for quantization of the classical Euclidean plane.


\newpage
\begin{thebibliography}{99}

\bibitem{pick}
J.~Agler, J.~E.~McCarthy:  {\it Pick Interpolation and Hilbert Function Spaces}, American Mathematical 
Society, 2002. 

\bibitem{SF} 
G.~Andrews, R.~Askey, R.~Roy: {\it Special Functions}, Encyclopedia of Mathematics and its Applications {\bf 71}, Cambridge University Press, 1999. 

\bibitem{barg}
V.~Bargmann:
{\it On a Hilbert space of analytic functions and its associated
integral transform.}~I, Commun. Pure Appl. Math. 
\textbf{14} (1961) 187--214.

\bibitem{berg}
S.~Bergman: {\it The Kernel Function and Conformal Mapping}, American Mathematical Society, 1950. 


\bibitem{bs}M.~S.~Birman, M.~Z.~Solomjak: {\it Spectral Theory of Self-Adjoint Operators in Hilbert Space}, Mathematics and Its Applications {\bf 5}, Springer, 1987. 

\bibitem{debranges}
L.~de Branges: {\it Hilbert Spaces of Entire Functions}, Prentice-Hall, 1968. 

\bibitem{debranges2}
L.~de Branges, J.~Rovnyak: {\it Square Summable Power Series}, Athena Series, Holt, Rinehart and Winston, 1966. 

\bibitem{connes} A.~Connes: {\it Noncommutative Geometry}, Academic 
Press, 1994.

\bibitem{MP} M.~{\Dj}ur{\dj}evich, P.~C.~Lucio~Pe\~na, {\it Geometric Structures on Quantum Hyperbolic Planes}, 
Q-Preprint Q-2020. 

\bibitem{coherent} M.~{\Dj}ur{\dj}evich, S.~B.~Sontz: {\it Coherent States for the Manin Plane
via Toeplitz Quantization}, J.~Math.~Phys.~{\bf 61}, 023502 (2020). 

\bibitem{qdunkl} M.~{\Dj}ur{\dj}evich, S.~B.~Sontz: {\it Dunkl Operators as Covariant Derivatives
in a Quantum Principal Bundle}, SIGMA \textbf{9} 040, 29 Pages (2013). 

\bibitem{qdunkl2} M.~{\Dj}ur{\dj}evich, S.~B.~Sontz: {\it Dunkl Operators for Arbitrary Finite Groups}, 
Adv.~Oper.~Theory \textbf{6}, 37 (2021). 

\bibitem{BHS}
G.~Gasper, M.~Rahman: {\it Basic Hypergeometric Series}, Encyclopedia of Mathematics and its Applications {\bf 96}, Cambridge University Press, 2004. 

\bibitem{rs1} M.~Reed, B.~Simon, {\it Methods of 
Modern Mathematical Physics, Vol. I. 
Functional Analysis}, Academic Press, 1972. 

\bibitem{pru} E.~Prugove\v{c}ki: {\it Quantum Geometry--A Framework for
Quantum General Relativity}, Kluwer Academic Publishers, 1992.

\bibitem{Schechter}
E.~Schechter: {\it Handbook of Analysis and its Foundations}, Academic Press (1996). 

\bibitem{schmudgen}
K.~Schm\"udgen: {\it The Moment Problem}, Springer, 2017. 

\bibitem{shepler}
A.~Shepler and S.~Witherspoon, 
{\it Poincar\'e-Birkhoff-Witt theorems},
in: Commutative algebra and non-commutative 
algebraic geometry, 
Vol. I, 259--290, 
Math. Sci. Res. Inst. Publ. 67, 
Cambridge Univ. Press, 2015. 

\bibitem{sbs1}
S.~B.~Sontz: 
{\it A Reproducing Kernel and Toeplitz Operators
in the Quantum Plane,} 
Commun. Math. \textbf{21} (2013) 137--160.

\bibitem{sbs2} 
S.~B.~Sontz: 
{\it Toeplitz Quantization for Non-commuting Symbol Spaces such as $SU_q (2)$},
Commun. Math.  \textbf{24} (2016) 43--69. 

\bibitem{Sti}
W. F.~Stinespring: {\it Positive Functions on C*-algebras}, Proceedings of the American Mathematical Society, \textbf{6} (1955) 211–216. 

\bibitem{slw-cat} S.~L.~Woronowicz: {\it Pseudospaces, Pseudogroups
and Potryagin Duality}, Proceedings of the International Conference of 
Mathematical Physics, Lausanne 1979, Lecture Notes in Physics {\bf 116},
407--412. 
\end{thebibliography}
\end{document}